%% file: main.tex
\documentclass[11pt,letterpaper]{article}

\usepackage{preamble}
\usepackage{authblk}

\title{Improved Learning with Structure:\\
Fine-Grained Complexity of Minimum Consistent Subset}

\author[1]{Robert Ganian%
  \thanks{\href{mailto:rganian@ac.tuwien.ac.at}{\texttt{rganian@ac.tuwien.ac.at}}}}
\author[2]{Manolis Vasilakis%
  \thanks{\href{mailto:emmanouil.vasilakis@dauphine.eu}{\texttt{emmanouil.vasilakis@dauphine.eu}}}}
\author[1]{Simon Wietheger%
  \thanks{\href{mailto:swietheger@ac.tuwien.ac.at}{\texttt{swietheger@ac.tuwien.ac.at}}}}

\affil[1]{Algorithms and Complexity Group, TU Wien, Austria}
\affil[2]{Universit\'{e} Paris Dauphine -- PSL, CNRS UMR7243, LAMSADE, France}

\date{}

\begin{document}

\maketitle

\begin{abstract}
\input{Sections/abstract}
\end{abstract}

\input{Sections/introduction}
\input{Sections/preliminaries}

\input{Sections/tw-algo}

\section{Establishing Optimality via Lower Bounds}

The aim of this section is to rule out any substantial improvements to \cref{thm:twalg}. For our first---and most challenging---lower bound, we rely on a careful reduction from a well-studied colored variant of the \textsc{Independent Set} problem.

\problemdef{\MCI (\MCIshort)}
{Graph $G$, a partition $V_1, \ldots, V_q$ of $V(G)$ for some $q,n\in \N$ such that $G[V_i]$ is a clique and $|V_i|=n$ for every $i\in[q]$.}
{Decide if $G$ has an independent set of size~$q$.}

Intuitively, each of the vertex subsets $V_1, \ldots, V_q$ is often viewed as having the same color, and the aim is then to find an independent set spanning each of the $q$ colors.
It is known that unless the Exponential Time Hypothesis fails, \MCI cannot be solved in time $f(q) n^{o(q)}$ for any computable function $f$~\cite{books/CyganFKLMPPS15}.

We now formalize the underlying construction.

\input{Sections/unweighted-reduction}

In combination with the aforementioned ETH-based lower bound for \MCI, \cref{lem:reduction:unweighted-c-pw-fvs-k} directly yields a matching lower bound for \cref{thm:twalg}.

\begin{theorem}\label{thm:whard:unweighted-c-pw-fvs-k}
    Even on instances where $c=2$, {\MCS} cannot be solved in time $f(k+\pw+\fvs)\cdot n^{o(k+\pw+\fvs)}$ under the \textup{ETH} for any computable function $f$.
\end{theorem}

Moreover, recall that \cref{cor:td} yields a running time of $2^{\bO(\td^2 + c\cdot\td)}\cdot n^{\bO(1)}$ for {\MCS} on unweighted graphs of treedepth $\td$.
By combining \cref{lem:reduction:unweighted-c-pw-fvs-k} with a careful reduction from \ThreeSAT\ to \MCI we show that, under the ETH, this dependence on $\td$ is essentially optimal.

\begin{theorem}\label{thm:treedepth:lower-bound}
    {\MCS} cannot be solved in time $2^{o(\td^2)}n^{\bO(1)}$ under the \textup{ETH}, even for $c=2$.
\end{theorem}

\begin{proof}
    We reduce {\ThreeSAT} to {\MCS} via {\MCI}.
    Let $\psi$ be a 3-CNF formula with variable set
    $X = \{x_1, \ldots, x_n\}$ and clause set
    $C = \{c_1, \ldots, c_m\}$.
    By padding with dummy clauses if necessary, we may assume that $q = \sqrt{m}$ is an integer.
    Partition $C$ into $q$ groups $C^1, \ldots, C^q$, each of size $q$.
    For every $i \in [q]$, let $V^i$ be the set of variables appearing in
    the clauses of $C^i$, so $|V^i| \le 3q$, and let
    $\Phi^i = \{\phi^i_1, \ldots, \phi^i_{N_i}\}$ be the set of partial
    assignments to $V^i$ that satisfy every clause in $C^i$. Note that $N_i \le 2^{3q}$.

    We construct an instance $(G, q)$ of {\MCI} by taking $V_i = \Phi^i$ for $i \in [q]$: each vertex class corresponds to one block of clauses, and each vertex represents a partial assignment that satisfies all clauses in its block. We make every $V_i$ a clique in $G$, so that any independent set contains at most one assignment per block. For every pair of partial assignments from different blocks that disagree on a shared variable, we add an edge between the corresponding vertices.

    Any independent set of $G$ of size $q$ thus consists of one assignment per block that are pairwise consistent and jointly satisfy every clause of $\psi$, witnessing satisfiability. Conversely, any satisfying assignment of $\psi$ restricts to one partial assignment per $V^i$; these are pairwise consistent and locally satisfying, so the corresponding vertices form an independent set in $G$.

    We may assume that each $V_i$ contains exactly $2^{3q}$ vertices, by duplicating vertices (with identical neighborhoods) within each class as needed. We may also assume that no $V_i$ is empty, as otherwise $\psi$ would be trivially unsatisfiable.

    The above reduces {\ThreeSAT} to an instance of {\MCI} with $q = \sqrt{m}$ vertex classes and $2^{3\sqrt{m}}$ vertices per class.
    Applying \cref{lem:reduction:unweighted-c-pw-fvs-k} produces an equivalent instance of {\MCS} with $c = 2$ and treedepth $\bO(q + \log(2^{3\sqrt{m}})) = \bO(\sqrt{m})$.
    The whole reduction can be carried out in $2^{\bO(\sqrt{m})}$ time.
    Hence an algorithm running in time $2^{o(\td^2)} N^{\bO(1)}$ on $N$-vertex instances of {\MCS} would yield a $2^{o(m)}$-time algorithm for {\ThreeSAT}, contradicting the ETH.
\end{proof}

\input{Sections/weighted-reduction}

\input{Sections/tree-reduction}

\input{Sections/vc-algos}

\input{Sections/conclusion}

% \input{Sections/improved-vc-algos}

\bibliographystyle{plainurl}
\bibliography{refs}

\end{document}

%% file: Sections/abstract.tex
Instance selection is a vital technique for mitigating the computational bottlenecks
of nearest-neighbor classification in large-scale supervised clustering.
A classical theoretical formulation of this objective is the \textsc{Minimum Consistent Subset (MCS)} problem.
While recent research has explored its complexity on unweighted graphs to uncover
structural boundaries of tractability, arbitrary metric spaces are much more accurately
modeled by (edge-)weighted graphs.

In this paper, we develop a comprehensive fine-grained complexity map of \textsc{MCS} on
both unweighted and weighted graphs. As our main result, we introduce a
$3^{c \cdot(\mathrm{tw}+1)}\cdot n^{\mathrm{tw}+\mathcal{O}(1)}$ algorithm for $n$-vertex $c$-colored
\textsc{MCS} instances on weighted graphs of treewidth $\mathrm{tw}$, substantially improving upon
the previous state-of-the-art algorithm for unweighted \textsc{MCS} on trees both in terms
of generality and running time. We complement this positive result with a series of lower
bounds that rule out asymptotic improvements to the running time for both
weighted and unweighted graphs under the Exponential Time Hypothesis (ETH). Moreover,
we improve the recent slightly superexponential vertex-cover based algorithm for unweighted
\textsc{MCS} (AAAI 2026) to a single-exponential one, and rule out further improvements to
subexponential running times under the ETH.
Together, our results strictly delineate the algorithmic boundaries of consistent subset selection
across diverse metric structures.

%% file: Sections/introduction.tex
\section{Introduction}

Clustering forms a foundation for a wide array of tasks across artificial intelligence.
Intuitively, given a collection of data points in a metric space, the goal of clustering is
to partition these into clusters so that proximate points are grouped into the same cluster.
While multiple methods for clustering have been considered to date, one that is prominently
employed in the machine learning setting is supervised clustering.
There, the process begins with a labeled training dataset---formally,
this can be viewed as a subset of points mapped via a coloring function, where each of the $c$
colors represents a distinct class or cluster.
The initial dataset is used to identify a set of cluster centers.
Once the cluster centers are computed, any new, unlabeled point is classified using a
nearest-neighbor approach: it receives the color of a closest cluster center.

However, scaling this nearest-neighbor framework to large datasets introduces severe
computational bottlenecks.
As the volume of data grows, identifying optimal cluster centers and querying nearest neighbors
becomes highly resource-intensive.
To mitigate this, a prominent line of research focuses on instance selection, which aims to
extract a compact subset of the training data that fully preserves
the original classification behavior.
A classical formulation of this objective is the \textsc{Metric Minimum Consistent Subset}
(\textsc{MMCS}) problem, introduced by Hart already in 1968~\cite{Hart}: Given a set of $n$
$c$-colored points in a metric space and a budget $k$, either compute a set $S$ of at most
$k$ points such that every point $p\not \in S$ has a nearest point $p'$ in $S$ with the same
color (i.e., among all points in $S$ with minimum distance to $p$, at least one must have the
same color as $p$) or determine that no such \emph{consistent set} $S$ of size at most $k$
exists.%
\footnote{We remark that here we state the decision version of \textsc{MMCS} to provide crisper
complexity-theoretic statements; every algorithm provided in this manuscript can also solve the
optimization version and output a set $S$ as witness for positive instances.}
%the optimization version---where we seek to minimize $k$---is equivalent to it.}.

Since its introduction, the computational complexity of \textsc{MMCS} has been studied in a variety
of settings and has been shown to be surprisingly resistant to classical computational approaches:
it is not only \NP-hard, but cannot even be solved in time
$f(k,d)\cdot n^{o(k^{1-\frac{1}{d}})}$, where $d$ is the dimensionality of the metric space
and $f$ is an arbitrary computable function, unless the Exponential Time Hypothesis (ETH)
fails~\cite{Chitnis22}.
Furthermore, \textsc{MMCS} is known to remain intractable even in Euclidean metric spaces~\cite{Wilfong,KhodamoradiKR18}.

\paragraph{Consistent Subsets in Graph Metrics.}
The lower bounds mentioned above motivated a more detailed study of \textsc{MMCS} in graph
metrics---that is, on instances where the data points are vertices of a graph $G$ and their
pairwise distances are the shortest-path distances in $G$.
While the problem remains \NP-hard on unweighted graphs~\cite{Wilfong}, this perspective
enables a fine-grained study of the problem's boundaries of tractability by taking into
account the structure of $G$.
A prominent recent result in this direction is that of Banik, Parthasarathi, Raman, Roy, and
Sahu~\cite{BanikPRRS26}, who obtained an algorithm solving the problem in time
$\vc^{\bO(\vc)}\cdot n^{\bO(1)}$ where $\vc$ is the minimum size of a vertex cover in an
unweighted graph $G$ (the so-called \emph{vertex cover number} of $G$).
Here, we follow their terminology and denote the restriction of \textsc{MMCS} to
unweighted graph metrics simply as \MCSshort.

Prior to that result, a number of recent research papers have obtained polynomial-time algorithms
for {\MCSshort} on various classes of unweighted graphs, including bicolored trees~\cite{fct/DeyMN21},
caterpillars~\cite{dam/DeyMN23}, and other more specialized classes~\cite{Manna25,Manna26}.
Notably, the former was generalized to a $2^{6c}\cdot n^{\bO(1)}$ algorithm on unweighted
trees~\cite{fsttcs/BanikDMMNMRRS24}, and the same paper also excluded polynomial-time
tractability for such trees.
%by establishing \NP-hardness when the number $c$ of colors.

\paragraph{Weighted Graph Metrics.}
While the body of previous work on graph metrics is impressive and non-trivial, all of the algorithmic
results mentioned above apply only to metric spaces that can be represented as unweighted graphs.
A natural way to model \emph{all} metric spaces is to work with edge-weighted graphs: every metric
space directly corresponds to a complete graph with edge weights encoding pairwise distances,
but whenever distances are transitive the model will naturally avoid direct edges (yielding the
same benefits as the unweighted graph model).
In view of the above, here we primarily focus on {\MCSshort} on edge-weighted graphs:

\problemdef{\WMCS~(\WMCSshort)}
{$n$-vertex graph $G$ with positive edge-weight function $w \colon E(G) \to \mathbb{Q}_{>0}$, integers $k,c \in \N$, and a vertex coloring $\col \colon V(G) \to [c]$.}
{Decide if there exists a consistent vertex set $S$ of size at most $k$.}

% \simon{I think we never use that edge weights are natural numbers, shall we instead have weights
% in $\mathbb{Q}$?}
% \manolis{Does this mess up Fullkerson or Dijkstra?
% Otherwise, I think that indeed we can just have positive weights in $\mathbb{Q}$.}
% \simon{You mean Floyd-Warshall for APSP? (ah, and dijsktra for weighted VC). Then yes, there shouldn't be any problem with rational weights }

It is worth noting that since {\WMCSshort} on edge-weighted complete graphs already captures arbitrary metric spaces, we cannot hope for efficient algorithms solving
{\WMCSshort} on dense graph classes (such as unit interval graphs~\cite{Manna25} or graphs of
bounded neighborhood diversity~\cite{BanikPRRS26}).
Nevertheless, there was previously no known obstacle that would prevent lifting the
recently developed algorithms for sparse graph classes to the weighted setting.

\paragraph{Our Contributions.}
The primary aim of this article is to obtain a comprehensive and detailed understanding of the
running time required to solve \WMCSshort.
However, as a by-product of our insights and results, we also substantially improve both the
upper and lower running time bounds for the previously studied \MCSshort.

In our main result, we show that {\WMCSshort} can be solved efficiently not only on trees, but
on all ``tree-like'' graphs when the number $c$ of colors is small.
Here, the notion of tree-likeness is captured by the well-established \emph{treewidth} $\tw$ of
the graph~\cite{RobertsonS84}: trees have a treewidth of $1$ while, e.g., an $n$-vertex complete
graph (which is far from tree-like) has a treewidth of $n-1$.
We prove:

\begin{restatable}{theorem}{twalg}\label{thm:twalg}
    {\WMCS} can be solved in time $3^{c\cdot(\tw+1)}\cdot n^{\tw+\bO(1)}$.
\end{restatable}

On a high level, the proof of \cref{thm:twalg} follows the well-established framework of dynamic programming over a nice
tree decomposition. However, the information that we store in the dynamic programming tables
does not concern only the part of the solution already in the subtree of the corresponding node, but
also the part that is outside of the subtree in the form of ``guessed'' information about a global solution.
Apart from that, in order to keep the exponential
dependence on $c$ under control, we make use of convolution-style techniques for efficiently combining tables
on the join nodes of a tree decomposition~\cite[Section~10]{books/CyganFKLMPPS15}.

Crucially, \cref{thm:twalg} has significant implications on the state of the art.
First, it implies a $9^{c} \cdot n^{\bO(1)}$ algorithm for {\MCSshort} on unweighted trees,
thus significantly improving over the previous $2^{6c} \cdot n^{\bO(1)}$ algorithm of Banik et al.~\cite{fsttcs/BanikDMMNMRRS24}.

Second, we complement \cref{thm:twalg} with a matching lower bound which excludes any fundamental
improvements of the result for \WMCSshort.
Notably, in \cref{thm:whard:weighted-vc} we use the ETH to exclude any algorithm solving
{\WMCSshort} in time $n^{o(\vc+k)}$, even on instances with only $2$ colors.
Since it is well-known that $\tw \leq \vc$, this means that the second term in the running time
of \cref{thm:twalg} is tight.
Among others, this implies that it is impossible to lift the known
$\vc^{\bO(\vc)}\cdot n^{\bO(1)}$ algorithm~\cite{BanikPRRS26} for {\MCSshort} to the
weighted setting.

Third, we show that the running time of \cref{thm:twalg} is essentially tight even when applied
in the easier unweighted setting of \MCSshort.
In particular, under the ETH we exclude:
\begin{enumerate}
    \item any $n^{o(\tw+k)}$ algorithm for {\MCSshort} (\cref{thm:whard:unweighted-c-pw-fvs-k}),
        which implies that the second exponential term is unavoidable. This lower bound relies on
        a highly non-trivial reduction, and the statement holds even when there are only $2$ colors
        and $\tw$ is replaced by more restrictive graph parameters such as \emph{pathwidth} and the
        \emph{feedback vertex number};

    \item any $2^{o(c)}\cdot n^{\bO(1)}$ (and, in fact, any $2^{o(n)}$) algorithm for {\MCSshort} even when restricted to trees of bounded depth (\cref{thm:nphard:trees}).
        To obtain our reduction, we both strengthen and significantly simplify the ideas used in previous work on \MCSshort~\cite[Section 3]{fsttcs/BanikDMMNMRRS24}.
        %        Our reduction is loosely inspired by previous work~\cite[Section 3]{fsttcs/BanikDMMNMRRS24}, which we both strengthen and significantly simplify.
        We thereby exclude any sublinear dependence on the first exponential term.
\end{enumerate}

Fourth, via a small adaptation of \cref{thm:twalg} we obtain a
$2^{\bO(c\cdot \td+\td^2)}\cdot n^{\bO(1)}$ algorithm for \MCSshort, where $\td$
stands for an alternative measure of tree-likeness called \emph{treedepth} (\cref{cor:td}).
Treedepth can be seen as an intermediate step between treewidth and the vertex
cover number: it is small on graphs that are structurally similar to a shallow tree.
One immediately notices that the parameter dependence on the treedepth in the exponent is,
unlike in the previous algorithms, quadratic.
We complement this with a fine-grained lower bound (\cref{thm:treedepth:lower-bound})
which shows that this quadratic dependence in \cref{cor:td} is unavoidable under the ETH.

Finally, with these new insights in hand we turn back towards the recent
$\vc^{\bO(\vc)}\cdot n^{\bO(1)}$ algorithm for \MCSshort~\cite{BanikPRRS26}.
We conclude by improving that result to achieve a single-exponential running time of
$5^\vc\cdot n^{\bO(1)}$ (\cref{thm:vc:unweighted-fpt}),
whereas a single-exponential dependency on $\vc$ can be seen to be optimal under the ETH due to a trivial reduction from \textsc{Dominating Set}~\cite[Section~2]{fsttcs/BanikDMMNMRRS24}.
%and match this with a new ETH lower bound excluding any subexponential dependence on $\vc$ (\cref{thm:vc:unweighted-eth-lower-bound}).
Using the very same idea, in \cref{thm:vc:weighted-xp} we obtain an
$n^{\bO(\vc)}$ algorithm for \WMCSshort, whose asymptotic optimality already follows from
\cref{thm:whard:weighted-vc}.
A \textbf{mindmap} summarizing our results on {\WMCSshort} and {\MCSshort} is provided in \cref{fig:results-overview}.

\begin{figure}
    \centering
    \input{graphics/results_overview}
    \caption{
        %An overview of the complexity of \textsc{(W)MCS} implied by our results.
        An overview of the relations between our algorithms and lower bounds; the latter
        assume the ETH.}
    \label{fig:results-overview}
\end{figure}
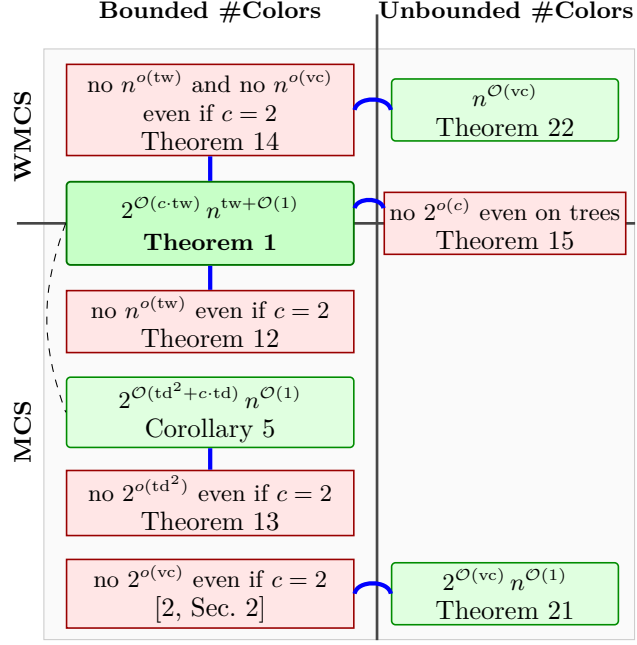

%  \todo[inline]{TODO: add table with all results}

\paragraph{Broader Related Work.}
There are many recent articles investigating primarily empirical aspects of supervised clustering---see,
e.g., the recent works on the topic~\cite{YangWZWZL025,clustering1,ZhangRWXDY25,clustering2}.
From a more foundational perspective, there is likewise a large body of works investigating how the
complexity of problems and tasks arising in AI is influenced by structural properties of inputs.
Apart from the already mentioned past works on \MCSshort, the graph parameters considered in this
article were utilized to identify algorithmic upper and lower bounds for a variety of other problems
in many other areas of AI research, such as computational social choice and machine learning---see,
e.g., the dedicated surveys~\cite{ChenHS26,Ganian26}.

%% file: graphics/results_overview.tex
% Overview of the paper's results, arranged in a 2x2 grid.
%
% Only four geometry parameters: \wA, \wB, \hA, \hB (column widths and
% row heights). All result boxes live inside a scope shifted to the
% horizontal dividing line (bgTL.south), so changing \hA / \hB keeps
% twalg centred on the divider and every other box at the same relative
% offset. Headers and row labels are anchored to the bg cells so they
% follow column widths and row heights too.

\begin{tikzpicture}[
    every node/.style={font=\footnotesize},
    algo/.style={
        draw, rounded corners=1.5pt, semithick,
        fill=green!12, draw=green!55!black,
        align=center, inner sep=2pt,
        minimum width=18mm, minimum height=8mm,
    },
    mainalgo/.style={
        algo,
        line width=0.7pt,
        fill=green!22, draw=green!45!black,
        font=\footnotesize\bfseries,
        minimum width=38mm, minimum height=11mm,
        inner sep=2.5pt,
    },
    lb/.style={
        draw, semithick,
        fill=red!10, draw=red!60!black,
        align=center, inner sep=2pt,
        minimum width=38mm, minimum height=7mm,
    },
    backbox/.style={draw=gray!40, fill=gray!4, inner sep=0pt,},
    header/.style={font=\footnotesize\bfseries},
    match/.style={ultra thick, blue},
]

    % ---- the four geometry knobs (mm) ----------------------------------
    \def\wA{44}\def\wB{34}\def\hA{55}\def\hB{23}

    % ---- background cells (anchor-chained so they always touch) -------
    \begin{scope}[on background layer]
        \node[backbox, minimum width=\wA mm, minimum height=\hB mm] (bgTL) at (0mm, 0mm) {};
        \node[backbox, anchor=north, minimum width=\wA mm, minimum height=\hA mm] (bgBL) at (bgTL.south) {};
        \node[backbox, anchor=west,       minimum width=\wB mm, minimum height=\hB mm] (bgTR) at (bgTL.east)       {};
        \node[backbox, anchor=north west, minimum width=\wB mm, minimum height=\hA mm] (bgBR) at (bgTL.south east) {};

      \draw[gray!55!black, line width=1pt]
        ($(bgTL.south west) + (-.35cm, 0)$) -- ($(bgTR.south east) + (0, 0)$);
     \draw[gray!55!black, line width=1pt]
        ($(bgTL.north east) + (0, .45cm)$) -- ($(bgBL.south east) + (0, 0)$);
    \end{scope}

    % ---- headers (anchored to the bg cells) ----------------------------
    \node[header, anchor=south]  at ($(bgTL.north) + (0, 2mm)$) {Bounded \#Colors};
    \node[header, anchor=south]  at ($(bgTR.north) + (0, 2mm)$) {Unbounded \#Colors};
    \node[header, rotate=90, anchor=south] at ($(bgTL.west) + (-0.3mm, 0)$) {\WMCSshort};
    \node[header, rotate=90, anchor=south] at ($(bgBL.west) + (-0.3mm, 0)$) {\MCSshort};

    % ---- all result boxes live in a scope whose origin is the divider
    %      line, so twalg at (0,0) always sits centred on it -------------
    \begin{scope}[shift={(bgTL.south)}]

    % col 1
    \node[mainalgo] (twalg) at (0mm, 0mm)
        {$2^{\bO(c\cdot \tw)}\, n^{\tw+\bO(1)}$\\[1pt]
         {\cref{thm:twalg}}};

    \node[lb] (lbwvc) at (0mm, 15mm)
        {no $n^{o(\tw)}$ and no $n^{o(\vc)}$\\
        even if $c=2$ \\
        {\small \cref{thm:whard:weighted-vc}}};

    \node[lb] (lbpwfvs) at (0mm, -13mm)
        {no $n^{o(\tw)}$ even if $c=2$\\
        {\small \cref{thm:whard:unweighted-c-pw-fvs-k}}};

    \node[algo, minimum width=38mm] (cortd) at (0mm, -25mm)
        {$2^{\bO(\td^2 + c\cdot\td)}\, n^{\bO(1)}$ \\ {\small \cref{cor:td}}};

    \node[lb] (lbtd) at (0mm, -37mm)
        {no $2^{o(\td^2)}$ even if $c=2$ \\ {\small \cref{thm:treedepth:lower-bound}}};

    \node[lb] (lbvc) at (0mm, -49mm)
        {no $2^{o(\vc)}$ even if $c=2$ \\ {\small \cite[Sec.~2]{fsttcs/BanikDMMNMRRS24}}};

    % col 2
    \node[algo, minimum width=30mm] (xpvc) at (39mm, 15mm)
        {$n^{\bO(\vc)}$ \\ {\small \cref{thm:vc:weighted-xp}}};

    \node[lb, minimum width=30mm,] (lbsubexp) at (39mm, 0mm)
        {no $2^{o(c)}$ even
        on trees\\ {\small \cref{thm:nphard:trees}}};
        %\\ {\small\cite{fsttcs/BanikDMMNMRRS24})}};

    \node[algo, minimum width=30mm] (fptvc) at (39mm, -49mm)
        {$2^{\bO(\vc)}\, n^{\bO(1)}$ \\ {\small \cref{thm:vc:unweighted-fpt}}};

    % ---- complementing pairs (dashed) ----------------------------------
    \draw[match] (lbwvc)    -- (twalg);
    \draw[match] ($(twalg.east) + (0, 2mm)$) to[bend left=80] ($(lbsubexp.west) + (0, 2mm)$);
    \draw[match] (lbpwfvs)  -- (twalg);
    \draw[match] (lbtd)     -- (cortd);
    \draw[match] (lbwvc.east) to[bend left=80] (xpvc.west);
    \draw[match] (lbvc.east)  to[bend left=80] (fptvc.west);
    \draw[dashed] (twalg.west) to[bend right=20] (cortd.west);

    \end{scope}

\end{tikzpicture}

%% file: Sections/preliminaries.tex
\section{Preliminaries}
We use standard graph-theoretic notation~\cite{books/Diestel25}.
For a graph $G$ we write $V(G)$ and $E(G)$ to denote its vertex and edge set,
respectively, while $n = |V(G)|$ and $m = |E(G)|$.
For a graph $G$ and vertices $v,w \in V(G)$, we let $\dist_G(v,w)$ denote their distance in $G$,
where $\dist_G(v,v)=0$; whenever $G$ is equipped with an edge-weight function, all distances in
$G$ are measured with respect to the edge weights.
Similarly, for a vertex $v$ and a vertex subset $S\subseteq V(G)$, we define
$\dist_G(v,S) = \min_{s \in S} \dist_G(v,s)$.
% Moreover, we let $N_{G,d}(v)$ denote the set of vertices at distance \emph{precisely} $d$
% from $v$ in $G$.
% We abbreviate $N_{G,1}(v)$ to $N_G(v)$ and omit $G$ from the subscript if it is clear from context.

%\paragraph{ETH.}
Our lower bounds are based on the \emph{Exponential Time Hypothesis} (ETH)~\cite{ImpagliazzoPZ01}, and specifically its well-known consequence that {\ThreeSAT} with $n$ variables
and $m$ clauses cannot be solved in time $2^{o(n+m)}$; see, e.g., the survey of Lokshtanov, Marx and Saurabh~\cite{LokshtanovMS11}.
%By the sparsification lemma~\cite{ImpagliazzoPZ01}, we may also assume that the input formula is sparse, that is, has linearly many clauses.
%For further background, we refer interested readers to the relevant survey on the topic~\cite{LokshtanovMS11}.

\paragraph{Consistent Subsets.}
For a positive integer $n \in \N$, let $[n] = \{1,2,\ldots, n\}$.
Given a graph $G$ and a coloring function $\col \colon V(G) \to [c]$ for some
$c \in \N$, for each $i \in [c]$ we define
$V_i = \setdef{v \in V(G)}{\col(v)=i}$ to be the vertices of color $i$.

Given a set $S \subseteq V(G)$, we say that a vertex $v$ is \emph{satisfied by $S$} if some
closest vertex of $S$ to $v$ has the same color as $v$; equivalently, if there is $u \in S \cap V_{\col(v)}$ such that 
$
    \dist_G(v, u) = \dist_G(v,S).
$
To provide a more formal underpinning for the definition of \WMCSshort in the Introduction, we call $S$ a \emph{consistent subset} of $G$ if every vertex of $G$ is satisfied by $S$.
In this case, we call all vertices in $S$ \emph{selected}.
% We can now state \WMCSshort in formal terms:
% \problemdef{{\WMCS} (\WMCSshort)}
% {Graph $G$ with an edge-weight function $w \colon E(G) \to \N$, $k \in \N$, and a vertex
% coloring function $\col \colon V(G) \to [c]$ for some $c \in \N$.}
% {Decide if $G$ has a consistent subset of cardinality at most $k$.}
{\MCS} (\MCSshort) is then  the restriction of {\WMCSshort} to the case where all edges have weight~$1$.
%Clearly, both problems are in $\NP$ as we can verify in polynomial time whether a given vertex subset is consistent.
% Thus, whenever we prove $\NP$-hardness for a subset of instances, we indeed obtain
% $\NP$-completeness. \todo{We never speak about NP hardness. Shall we drop that sentence? Or mention here that all our ETH reductions also imply NP hardness? }

\paragraph{Graph Measures.}
The \emph{treewidth} $\tw(G)$ of a graph~$G$ is a well-established measure of how ``tree-like''
it is; for instance, trees have a treewidth of $1$, while an $n$-vertex complete graph has
treewidth $n-1$. The notion is defined via the concept of (nice) tree decompositions.

Formally, a \emph{nice tree decomposition} of $G$ is a pair $(T, \{X_t\}_{t\in V(T)})$, where
$T$ is a tree (whose vertices are called \emph{nodes}) rooted at a node $r$ and
$\{X_t\}_{t\in V(T)}$ assigns to each node~$t$ a set $X_t \subseteq V(G)$ such that:
\begin{itemize}
    \item For every $vw \in E(G)$, there is a node $t$ such that $\{v, w\} \subseteq X_t$.
    \item For every vertex $v \in V(G)$, the set of nodes $t$ satisfying $v \in X_t$ forms a
        subtree of $T$.
    \item $|X_t| = 0$ if $t$ is a leaf of $T$ or $t=r$.
    \item There are only three kinds of non-leaf nodes in $T$:
        \begin{itemize}
            \item \emph{introduce}: a node $t$ with exactly one child~$t'$ such that
                $X_t = X_{t'} \cup \{v\}$ for a vertex $v \notin X_{t'}$.
            \item \emph{forget}: a node $t$ with exactly one child $t'$ such that
                $X_t = X_{t'} \setminus \{v\}$ for a vertex $v \in X_{t'}$.
            \item \emph{join}: a node $t$ with two children $t_1, t_2$ such that
                $X_t = X_{t_1} = X_{t_2}$.
        \end{itemize}
\end{itemize}

We call each set $X_t$ a \emph{bag}.
The width of a nice tree decomposition $(T, \{X_t\}_{t\in V(T)})$ is the size of the largest
bag $X_t$ minus $1$, and the \emph{treewidth of $G$}, denoted by $\tw(G)$, is the minimum width
of a nice tree decomposition of $G$.
For further details on treewidth, we refer the reader to one of the relevant books~\cite[Section~7]{books/CyganFKLMPPS15}.
We remark that there is a broad range of algorithms for computing an optimal or near-optimal
tree decomposition of a graph $G$; for our purposes, it suffices to employ the classical exact $\bO(n^{\tw(G)+2})$ algorithm~\cite{arnborg1987complexity}.

%algorithm of Arnborg, Corneil, and Proskurowski, which runs in time $\bO(n^{\tw(G)+2})$~\cite{arnborg1987complexity}.

The \emph{pathwidth} $\pw(G)$ of $G$ is defined analogously to treewidth, but with the
distinction that there are no \emph{join} nodes.

The \emph{vertex cover number} $\vc(G)$ of a graph $G$ is defined as the minimum size of a vertex
cover of $G$, that is, the minimum size of a vertex subset $X$ such that each edge of $G$ has at
least one endpoint in $X$.
It is well-known that a vertex cover of minimum size can be computed in time
$\bO(2^{\vc(G)}\cdot |V(G)|)$~\cite{books/CyganFKLMPPS15}.

The \emph{treedepth} $\td(G)$ of a graph $G$ is a graph parameter that characterizes how close a
graph is to a forest of stars.
While several equivalent definitions of the notion exist~\cite{sparsity}, for our algorithmic
applications it suffices to use the \emph{elimination distance} characterization: the treedepth
of a one-vertex graph is $1$, the treedepth of a disconnected graph is the maximum treedepth
among its connected components, and the treedepth of a connected graph with more than one vertex
is $1$ plus the minimum treedepth of a graph $G'$ obtained by removing a single vertex from $G$.
%: the $\td(G)$ of a graph $G$ is the minimum depth of a recursion tree in which, at each step, one deletes a single vertex and recurses independently on the newly formed connected components until each such component is a singleton.
We note the well-known facts~\cite{sparsity,GanianO18} that (1) an edge-unweighted graph cannot
contain any path of length greater than $2^{\td(G)}-1$, and
(2) $\vc(G)+1\geq \td(G) \geq \tw(G)$.

%% file: Sections/tw-algo.tex
\section{A Treewidth-Based Algorithm for \WMCSshort}

In this section we present our main algorithmic result.
Similarly to virtually all tree-width based algorithms, our algorithm also proceeds via dynamic programming over a nice tree
%The algorithm follows the standard dynamic-programming framework over a nice tree
decomposition---however, in our case the information stored at a bag is global. In particular, for every bag vertex we
will be storing its ``promised'' distance to the final solution, not only to the part of the solution
already seen in the current subtree.
We will additionally store the colors that will appear at that distance, together with the subset of
these colors that has already been witnessed inside the processed part of the decomposition.
This is what allows us to ``safely'' process forget nodes: a vertex can be forgotten even if some of its closest
witnesses lie outside the current subtree, as their effect is still represented by the
remaining bag promises.

At join nodes, the two children will be contributing two already-witnessed color sets. To compute the new records, we will view the update operation as the computation of a \emph{cover product} and use this to maintain the desired single-exponential dependence on the number of colors.
%At join nodes, the two children contribute two already-witnessed color sets whose union becomes the parent set; we compute this operation by a cover product, which keeps the exponential dependence on the number of colors mild.

\twalg*

\begin{proof}
    Let $(G, w, k, \col)$ be an instance of \WMCS.
    It is enough to describe the algorithm for connected graphs.
    If $G$ is disconnected, then each connected component has to be satisfied independently,
    and the optimum size is the sum of the optimum sizes of the components.

    As a first step, we compute the shortest distance between all pairs of vertices in $G$ via
    the Floyd--Warshall algorithm, which runs in $\bO(n^3)$ time even for binary-encoded weights~\cite{floyd1962algorithm,warshall1962theorem}.
    From this, derive a set $A_v = \setdef{\dist(v,u)}{u \in V(G)}$ for every vertex $v \in V(G)$,
    and note that $|A_v| \le n$ and $0 \in A_v$ since $\dist(v,v)=0$.%
    \footnote{Throughout the proof, all distances are with respect to graph $G$, thus we omit the subscript from $\dist$.}

    We compute a nice tree decomposition $(T, \{X_t\}_{t\in V(T)})$ of $G$ of width $\tw$, using, e.g.,
    the algorithm of Arnborg, Corneil, and Proskurowski~\cite{arnborg1987complexity} that runs in time $n^{\tw + \bO(1)}$.
    Let $r$ denote the node that is the root of the decomposition, where $X_r = \varnothing$.
    For a node $t \in V(T)$, we denote by $V_t \subseteq V(G)$ the union of the bags of $t$ and all its descendants,
    and we put $Y_t \coloneq V_t \setminus X_t$.
    We will repeatedly use the following standard separation property of tree decompositions:
    if a vertex lies in one side of a decomposition node and another vertex lies outside that
    side, then every path between them meets the bag of the node.
    In particular, in an introduce node $X_t=X_{t'}\cup\{v\}$ every path from $v$ to
    $Y_t=Y_{t'}$ meets $X_{t'}$, and in a join node every path between the two child sides
    meets the common bag.

    \paragraph{Signatures.}
    Fix a node $t \in V(T)$ and write $X_t = \{v_1,\ldots,v_b\}$ with $b \le \tw+1$.
    A \emph{signature} $\sigma$ at $t$ is a tuple
    \[
        (s, \, (f_i)_{i\in[b]}, \, (R_i)_{i\in[b]}, \, (C_i)_{i\in[b]}),
    \]
    such that $s\in\{0,\ldots,k\}$ and for all $i\in[b]$,
    $f_i \in A_{v_i}$ and $C_i \subseteq R_i \subseteq [c]$, while $\col(v_i) \in R_i$.
    Moreover, if $f_i = 0$, then $C_i = \varnothing$ and $R_i = \{\col(v_i)\}$.
    Let $\Sigma_t$ denote the set of all signatures at node $t$.
    We have
    \[
        |\Sigma_t| \le (k+1) \cdot \prod_{i=1}^{b} |A_{v_i}| \cdot (3^{c})^{b}
        \le (k+1) \cdot n^{b} \cdot 3^{cb},
    \]
    since each color is either in $C_i$, in $R_i\setminus C_i$, or outside $R_i$.

    The entries $(f_i,R_i)$ are promises about the final closest distance of $v_i$ to the
    solution and the colors seen at this distance, while $C_i$ records the part of this
    color set that is already witnessed by selected vertices inside $Y_t$.
    For $S_t \subseteq Y_t$ and a signature $\sigma$, define, for every $u \in V_t$,
    \[
        \delta_\sigma(u,S_t)
        \coloneq
        \min\left\{
            \dist(u,S_t),
            \min_{i\in[b]} (\dist(u,v_i)+f_i)
        \right\}.
    \]
    Here $\dist(u,\varnothing)=\infty$, and the minimum over an empty bag is also $\infty$.
    Also define the set of colors available to $u$ at this promised closest distance as
    \[
        \begin{aligned}
        \Gamma_\sigma(u,S_t)
        \coloneq{}&
        \setdef{\col(x)}{x\in S_t,\ \dist(u,x)=\delta_\sigma(u,S_t)}
        \\
        &{}\cup
        \bigcup_{\substack{i\in[b]\\ \dist(u,v_i)+f_i=\delta_\sigma(u,S_t)}} R_i.
        \end{aligned}
    \]
    We say that the entries of two bag vertices $v_i,v_j$ are \emph{coherent from
    $v_j$ to $v_i$} if
    \[
        \dist(v_i,v_j)+f_j\ge f_i,
    \]
    and, in the case of equality, $R_j\subseteq R_i$ and $C_j\subseteq C_i$.

    We say that $S_t \subseteq Y_t$ \emph{witnesses} $\sigma \in \Sigma_t$ at $t$ if the following conditions hold:
    \begin{enumerate}
        \item\label{item:tw-witness-size} $|S_t|=s$;
        \item\label{item:tw-witness-coherence} all ordered pairs of bag vertices are coherent;
        \item\label{item:tw-witness-boundary-distance} for every $i\in[b]$, $\dist(v_i,S_t)\ge f_i$;
        \item\label{item:tw-witness-processed-colors} for every $i\in[b]$,
        \[
            C_i=\setdef{\col(x)}{x\in S_t,\ \dist(v_i,x)=f_i};
        \]
        \item\label{item:tw-witness-satisfaction} for every $u\in Y_t$, $\col(u)\in\Gamma_\sigma(u,S_t)$.
    \end{enumerate}
    Equality in Condition~\ref{item:tw-witness-boundary-distance} is not required, as the
    vertices realizing the final distance $f_i$ may all lie outside of $Y_t$.
    We call $\sigma$ \emph{valid} at $t$ if there exists a set $S_t\subseteq Y_t$ that
    witnesses $\sigma$ at $t$.
    The dynamic program stores a Boolean table $\DPtable_t$ over $\Sigma_t$.
    The proof below shows that every true entry is valid, and that every global solution
    induces true entries at all nodes.

    \paragraph{Leaf Node.}
    If $t$ is a leaf, then $X_t=Y_t=\varnothing$.
    For $\sigma = (s,\varnothing,\varnothing,\varnothing) \in \Sigma_t$,
    we set
    \[
        \DPtable_t[(s,\varnothing,\varnothing,\varnothing)] =
        \begin{cases}
            \true,  & s=0,\\
            \false, & s\in[k].
        \end{cases}
    \]

    \paragraph{Introduce Node.}
    Let $t$ be an introduce node with child $t'$ and $X_t=X_{t'}\cup\{v\}$.
    We inspect every signature $\sigma\in\Sigma_t$.
    Let $\sigma^{\downarrow}$ be the restriction of $\sigma$ to the old bag $X_{t'}$ (that is, $s$ is the same as in $\sigma$ and the three sets each drop their entry for $v$),
    and let $(f_v,R_v,C_v)$ be the entries of $\sigma$ for $v$.
    We set $\DPtable_t[\sigma]=\false$ if $\DPtable_{t'}[\sigma^{\downarrow}]=\false$, or if one of
    the following tests fails.
    First, $v$ and every vertex of $X_{t'}$ must be coherent in both directions.
    Second, the processed colors seen by $v$ must be exactly
    \begin{equation}\label{eq:tw-introduce-cv}
        C_v =
        \bigcup_{\substack{v_i\in X_{t'}\\ \dist(v,v_i)+f_i=f_v}} C_i.
    \end{equation}
    If all tests pass, we set $\DPtable_t[\sigma]=\DPtable_{t'}[\sigma^{\downarrow}]$.

    \paragraph{Forget Node.}
    Let $t$ be a forget node with child $t'$ and $X_{t'}=X_t\cup\{v\}$.
    We initialize $\DPtable_t$ to $\false$.
    Consider a signature $\sigma'\in\Sigma_{t'}$ with $\DPtable_{t'}[\sigma']=\true$,
    and write $(f_v,R_v,C_v)$ for the entries of the forgotten vertex.
    We first require that $v$ and every remaining bag vertex are coherent with respect to $\sigma'$ in both directions.
    Next we finalize the colors promised for $v$ by checking that
    \begin{equation}\label{eq:tw-forget-rv}
        R_v =
        \begin{cases}
            \{\col(v)\}, & f_v=0,\\[1mm]
            C_v \cup
            \displaystyle\bigcup_{\substack{v_i\in X_t\\ \dist(v,v_i)+f_i=f_v}} R_i,
            & f_v>0.
        \end{cases}
    \end{equation}
    If any of these checks fails, $\sigma'$ produces no parent signature.
    Otherwise, we obtain a parent signature by dropping the entries of $v$ from $\sigma'$
    and updating the remaining fields as follows.
    If $f_v>0$, the size and all remaining $C_i$ entries stay unchanged.
    If $f_v=0$, then $v$ is selected when it enters $Y_t$: we increase the size $s$ by one and,
    for every $v_i\in X_t$ with $\dist(v_i,v)=f_i$, we add $\col(v)$ to $C_i$.
    If the resulting tuple is a signature $\sigma\in\Sigma_t$ with size at most $k$, this is
    the parent signature produced from $\sigma'$, and we set $\DPtable_t[\sigma]=\true$.

    \paragraph{Join Node.}
    Let $t$ be a join node with children $t_1,t_2$ and
    $X_t=X_{t_1}=X_{t_2}$.
    We combine pairs of signatures $\sigma_1\in\Sigma_{t_1}$ and
    $\sigma_2\in\Sigma_{t_2}$ with true table entries that agree on all $(f_i,R_i)$.
    The candidate parent tuple produced from $(\sigma_1,\sigma_2)$ keeps these common
    $(f_i,R_i)$ values, has size $s=s_1+s_2$, and has
    \begin{equation}\label{eq:tw-join-ci}
        C_i=C_i^{(1)}\cup C_i^{(2)}
        \qquad\text{for every }v_i\in X_t.
    \end{equation}
    If this tuple is a signature $\sigma\in\Sigma_t$ and $s\le k$, we set $\DPtable_t[\sigma]=\true$.

    \paragraph{Correctness.}
    \begin{claim}[Soundness]\label{claim:tw-soundness}
        If $\DPtable_t[\sigma]=\true$, then $\sigma$ is valid.
    \end{claim}
    \begin{claimproof}
        We argue by induction over the nice tree decomposition.
        Let $\sigma \in \Sigma_t$ be a signature with $\DPtable_t[\sigma]=\true$.
        The case where $t$ is a leaf node is immediate.
        Assume that the claim holds for all children of $t$.

        Consider an introduce node $t$ with child $t'$ and $X_t=X_{t'}\cup\{v\}$.
        Let $\sigma^{\downarrow}$ be the restriction of $\sigma$ to $X_{t'}$,
        and let $(f_v,R_v,C_v)$ be the entries of $\sigma$ for $v$.
        It holds that $\DPtable_{t'}[\sigma^{\downarrow}]=\true$,
        so by induction there exists a set $S \subseteq Y_{t'} = Y_t$ that witnesses
        $\sigma^{\downarrow}$.
        Since every path from $v$ to $Y_t$ crosses $X_{t'}$,
        it holds that
        \begin{align*}
          \dist(v,S) &= \min_{v_i\in X_{t'}} \left(\dist(v,v_i)+\dist(v_i,S)\right)\\
                    &\ge \min_{v_i\in X_{t'}} \left(\dist(v,v_i)+f_i\right)\\
                    &\ge f_v. \tag{I}\label{eq:tw-introduce-distance-bound}
        \end{align*}
        Here the first inequality is Condition~\ref{item:tw-witness-boundary-distance} for
        the child witness, and the second is coherence from old bag vertices to $v$.
        By \cref{eq:tw-introduce-distance-bound}, no processed vertex in $S$ is closer to $v$ than
        $f_v$.
        A processed vertex $x\in S$ is at distance $f_v$ from $v$ exactly when some shortest
        $v$--$x$ path crosses an old bag vertex $v_i$ such that
        $\dist(v,v_i)+f_i=f_v$ and $\dist(v_i,x)=f_i$.
        Hence \cref{eq:tw-introduce-cv} is exactly
        Condition~\ref{item:tw-witness-processed-colors} for the introduced vertex.
        Adding the boundary promise of $v$ cannot make any vertex of $Y_t$ lose
        satisfaction: if a shortest route from such a vertex to the promise of $v$ first
        crosses an old bag vertex $v_i$, then the promise of $v_i$ is no farther; in the
        tight case, coherence gives $R_v\subseteq R_i$.
        The remaining conditions of a witness are inherited from the child, together with
        the coherence tests involving $v$.
        Thus the same set $S$ witnesses the parent signature $\sigma$.

        Consider a forget node $t$ with child $t'$ and $X_{t'}=X_t\cup\{v\}$.
        Since $\DPtable_t[\sigma]=\true$, the forget transition produced $\sigma$ from some
        child signature $\sigma'\in\Sigma_{t'}$ with $\DPtable_{t'}[\sigma']=\true$.
        By induction, let $S'$ witness this child signature $\sigma'$.
        Let $(f_v,R_v,C_v)$ be the entries of $\sigma'$ for $v$.
        If $f_v=0$, set $S=S'\cup\{v\}$; otherwise set $S=S'$.
        For every remaining bag vertex $v_i$,
        Condition~\ref{item:tw-witness-boundary-distance} for the child witness gives
        $\dist(v_i,S')\ge f_i$.
        If $v$ is added to the processed set, coherence from $v$ to $v_i$ gives
        $\dist(v_i,v)\ge f_i$ as well.
        If $v$ is added to $S$, the update adds $\col(v)$ exactly to those $C_i$ for which
        $\dist(v_i,v)=f_i$; otherwise the processed contribution to the $C_i$ sets is
        unchanged.
        Hence Conditions~\ref{item:tw-witness-boundary-distance} and
        \ref{item:tw-witness-processed-colors} hold for the parent witness.

        It remains to check satisfaction after the promise of $v$ is removed from the
        boundary.
        For the newly forgotten vertex $v$, \cref{eq:tw-forget-rv} says that every color in
        $R_v$ at distance $f_v$ is represented either by a processed vertex in $S'$ or by a
        remaining bag promise tight for $v$.
        Since every signature satisfies $\col(v)\in R_v$, this gives
        $\col(v)\in\Gamma_\sigma(v,S)$.
        Now take an older vertex $u\in Y_{t'}$.
        If a child witness for $\col(u)$ did not use the promise of $v$, the same witness is
        still present in the parent.
        Otherwise $\dist(u,v)+f_v=\delta_{\sigma'}(u,S')$ and $\col(u)\in R_v$.
        By finalization, this color is supplied either by some $x\in S'$ with
        $\dist(v,x)=f_v$, or by a remaining bag vertex $v_i$ with
        $\dist(v,v_i)+f_i=f_v$ and $\col(u)\in R_i$.
        In the first case $x$ is at distance at most $\dist(u,v)+f_v$ from $u$; in the second
        case the promise through $v_i$ is at distance at most $\dist(u,v)+f_v$.
        Since the parent only removes the boundary option of $v$ (or replaces it by the
        selected vertex $v$ when $f_v=0$), this replacement is at the new promised closest
        distance.
        Thus the parent signature still supplies $\col(u)$ at that distance.
        Therefore $S$ witnesses $\sigma$.

        Finally, consider a join node.
        Since $\DPtable_t[\sigma]=\true$, the join transition produced $\sigma$ from
        compatible child signatures $\sigma_1\in\Sigma_{t_1}$ and
        $\sigma_2\in\Sigma_{t_2}$ with true table entries.
        By induction, let $S_1$ and $S_2$ witness $\sigma_1$ and $\sigma_2$, respectively,
        and set $S=S_1\cup S_2$; note that $S_1 \cap S_2 = \varnothing$.
        The common $(f_i,R_i)$ promises agree, and \cref{eq:tw-join-ci} gives the exact
        processed colors seen from each bag vertex.
        Pairwise coherence is preserved as well: the distance and $R$-set conditions are
        shared by the two children, and every tight inclusion $C_j\subseteq C_i$ holds after
        taking the corresponding unions.

        It remains to justify satisfaction.
        Let $u\in Y_{t_1}$; the other case is symmetric.
        For every $x\in S_2$, a shortest path from $u$ to $x$ meets some bag vertex
        $v_i\in X_t$.
        Since $S_2$ witnesses $\sigma_2$, we have $\dist(v_i,x)\ge f_i$, and therefore
        \[
            \dist(u,x)\ge \dist(u,v_i)+f_i\ge \delta_{\sigma_1}(u,S_1).
        \]
        Thus adding $S_2$ cannot decrease the promised closest distance of $u$ below the
        distance already present in the first child; in fact
        $\delta_\sigma(u,S)=\delta_{\sigma_1}(u,S_1)$.
        The colors certifying $u$ in $\sigma_1$ are still available in the parent, because
        the parent keeps the same $(f_i,R_i)$ and contains $S_1$.
        Hence all vertices satisfy Condition~\ref{item:tw-witness-satisfaction}.
        Thus $S$ witnesses $\sigma$.
    \end{claimproof}

    \begin{claim}[Completeness]\label{claim:tw-completeness}
        Let $S \subseteq V(G)$ be a consistent subset of size at most $k$.
        For every node $t$, define $\sigma \in \Sigma_t$ as follows:
        $S_t = S\cap Y_t$, $s=|S_t|$, and for each $v_i\in X_t$,
        \begin{align*}
            f_i&=\dist(v_i,S),\\
            R_i&=\setdef{\col(x)}{x\in S,\ \dist(v_i,x)=f_i},\\
            C_i&=\setdef{\col(x)}{x\in S_t,\ \dist(v_i,x)=f_i}.
        \end{align*}
        Then the dynamic program sets $\DPtable_t[\sigma]=\true$.
    \end{claim}
    \begin{claimproof}
        The displayed tuple is a signature in $\Sigma_t$:
        consistency of $S$ gives $\col(v_i)\in R_i$; if $f_i=0$, then $v_i\in S$ and,
        since bag vertices are not in $Y_t$, we have
        $C_i=\varnothing$ and $R_i=\{\col(v_i)\}$.
        Coherence follows from the triangle inequality for the global distances to $S$,
        with the tight color inclusions following from the same equality case.
        The set $S_t$ witnesses this signature.
        Indeed, Conditions~\ref{item:tw-witness-size}--\ref{item:tw-witness-processed-colors}
        hold by definition.
        For the satisfaction condition, take $u\in Y_t$ and let $x\in S$ be a closest
        selected vertex of color $\col(u)$, which exists because $S$ is consistent.
        If $x\in X_t$, then the bag promise for $x$ itself has value $0$.
        If $x\notin V_t$, then every path from $u$ to $x$ leaves $V_t$ through some bag
        vertex $v_i$, and $f_i\le \dist(v_i,x)$.
        Thus, unless $x\in S_t$, some bag promise gives a route of length at most
        $\dist(u,x)=\dist(u,S)$ and contains the color $\col(x)$ in the tight case.
        Conversely, each processed vertex and each bag promise used in
        $\delta_\sigma(u,S_t)$ is backed by a vertex of $S$, so
        $\delta_\sigma(u,S_t)\ge \dist(u,S)$.
        Hence $\delta_\sigma(u,S_t)=\dist(u,S)$.
        If $x\in S_t$, then $\col(x)$ appears directly in $\Gamma_\sigma(u,S_t)$; otherwise
        it appears through a tight bag promise.
        Thus $\col(u)\in\Gamma_\sigma(u,S_t)$.
        We prove by induction that the transition at every node preserves the induced
        signature.
        The leaf case is trivial.

        At an introduce node, let $\sigma$ be the signature induced by $S$ at $t$ and
        let $\sigma'$ be the signature induced by $S$ at the child $t'$.
        Then $\sigma'$ is precisely the restriction of the parent signature $\sigma$ to the
        old bag.
        Since every path from the introduced vertex $v$ to $Y_t$ crosses $X_{t'}$, the
        selected vertices of $S_t$ at distance $f_v$ from $v$ are precisely accounted for
        by old bag vertices $v_i$ with $\dist(v,v_i)+f_i=f_v$ and processed witnesses
        counted in $C_i$.
        Hence the equality defining $C_v$ holds.
        The coherence checks are just the triangle inequality applied to the global closest
        distances in $S$, so the introduce transition accepts.

        At a forget node, let $\sigma'$ be the signature induced by $S$ at the child $t'$,
        and let $\sigma$ be the signature induced by $S$ at $t$.
        The child signature $\sigma'$ passes the same coherence checks.
        If $f_v=0$, then $v\in S$ and the transition adds $v$ to the processed set;
        otherwise $v\notin S$ and the size is unchanged.
        Every closest selected vertex to $v$ is either already in $Y_{t'}$ or is reached
        through a remaining bag vertex $v_i$ with $\dist(v,v_i)+f_i=f_v$; hence the
        finalization equation for $R_v$ holds and the transition applied to $\sigma'$
        produces exactly the parent induced signature $\sigma$.

        At a join node, let $\sigma_1,\sigma_2$ be the signatures induced by $S$ at the two
        children and let $\sigma$ be the signature induced at $t$.
        The induced selected sets on the two child sides are disjoint and
        their union is $S\cap Y_t$.
        The $(f_i,R_i)$ values are global and therefore agree on both sides, while the
        processed color sets satisfy $C_i=C_i^{(1)}\cup C_i^{(2)}$.
        Thus the join transition applied to $(\sigma_1,\sigma_2)$ produces the induced
        parent signature $\sigma$.
    \end{claimproof}

    At the root $r$ we have $X_r=\varnothing$ and $Y_r=V(G)$.
    Hence the only information in a root signature is the size $s$, and soundness says that
    a true root entry yields a consistent subset of size $s$.
    Conversely, completeness shows that every consistent subset of size at most $k$ induces
    a true root entry.
    We accept iff some root entry with $s\le k$ is true.

    \paragraph{Runtime.}
    For a bag of size $b\le\tw+1$, the number of signatures is at most
    \[
        (k+1)n^b3^{cb}\le 3^{c(\tw+1)}n^{\tw+\bO(1)}.
    \]
    Leaf, introduce, and forget transitions inspect a polynomial number of conditions per
    relevant signature.

    \begin{claim}[Cover product]\label{claim:tw-join-cover-product}
        For a fixed join node and a fixed choice of $(f_i,R_i)_{i\in[b]}$, all compatible
        parent entries for this choice can be computed in
        $2^{\sum_i |R_i|}\cdot n^{\bO(1)}$ time.
    \end{claim}
    \begin{claimproof}
        Let
        \[
            \Omega=\bigcup_{i\in[b]}(\{i\}\times R_i).
        \]
        A tuple $(C_i)_{i\in[b]}$ with $C_i\subseteq R_i$ is equivalently a subset
        $C\subseteq\Omega$.
        For each child $\ell\in\{1,2\}$ and each size $p$, let
        $A_\ell^p(C)\in\{0,1\}$ indicate whether the child has a true entry with
        size $p$ and encoded color tuple $C$.
        For fixed values of $p$ and $q$, the compatible pairs of child entries are counted by
        the cover product
        \[
            H_{p,q}(C)=
            \sum_{C_1\cup C_2=C}
            A_1^p(C_1) \cdot A_2^q(C_2).
        \]
        This is a cover product rather than a subset convolution because $C_1$ and $C_2$
        need not be disjoint: the same color can be witnessed at the same distance from
        both child sides.
        A parent entry of size $s$ and encoded color tuple $C$ is true exactly when
        $H_{p,q}(C)>0$ for some $p,q$ with $p+q=s$.

        By the standard fast cover-product algorithm
        \cite[Theorem~10.14]{books/CyganFKLMPPS15}, all values of $H_{p,q}$ for fixed
        $p,q$ can be computed in $2^{|\Omega|}\cdot |\Omega|^{\bO(1)}$ arithmetic
        operations.
        We do this for all $p,q\in\{0,\ldots,k\}$.
        Since $k \le n$, the guesses of $p$ and $q$ contribute only a polynomial factor.
        Hence the total time for this fixed choice of $(f_i,R_i)_{i\in[b]}$ is
        $2^{|\Omega|}\cdot n^{\bO(1)}$.
    \end{claimproof}

    Summing the bound of \cref{claim:tw-join-cover-product} over all choices of the $R_i$ gives
    \[
        \prod_{i=1}^{b}
        \sum_{\substack{R_i\subseteq[c]\\ \col(v_i)\in R_i}}2^{|R_i|}
        =
        (2\cdot 3^{c-1})^b
        \le 3^{cb}.
    \]
    The choices of the distance values contribute a factor of at most $n^b$.
    Including the choices of the distances and the $\bO(n)$ nodes of the decomposition,
    the total running time is
    $3^{c(\tw+1)}\cdot n^{\tw+\bO(1)}$.
\end{proof}

We leverage \cref{thm:twalg} to obtain the aforementioned treedepth-based algorithm for \MCSshort.
\begin{corollary}\label{cor:td}
    {\MCS} can be solved in time $2^{\bO(\td^2 + c\cdot\td)} \cdot n^{\bO(1)}$ on graphs of treedepth $\td$.
\end{corollary}

\begin{proof}
    We first compute a tree witnessing the elimination distance of $G$ of depth $\td$ in time
    $2^{\bO(\td^2)} n^{\bO(1)}$~\cite{esa/NadaraPS22}.
    Given such a tree, one can construct a nice tree decomposition of $G$ of width at most $\td-1$.
    We apply the algorithm of \cref{thm:twalg} on the unweighted input graph $G$.
    Every two vertices in the same connected component of $G$ are at distance at most $2^{\td}-1$,
    so $|A_v| \le 2^{\td}$ for every $v \in V(G)$.
    Replacing the previous $|A_v| \le n$ estimate,
    the running time of \cref{thm:twalg} becomes
    $3^{c \cdot \td}\cdot (2^{\td})^{\td} \cdot n^{\bO(1)}$.
    This completes the proof.
\end{proof}

%% file: Sections/unweighted-reduction.tex
\begin{lemma}\label{lem:reduction:unweighted-c-pw-fvs-k}
    There is a yes- and no-instance preserving polynomial-time reduction which, given an instance of {\MCI} with $q$ vertex classes of $n$ vertices each returns an instance $(H,k,\col)$ of {\MCS} with $c=2$, $k + \pw(H) + \fvs(H) \in \bO(q)$ and $\td(H) \in \bO(q+\log n)$.
\end{lemma}

\begin{proof}
    Let $(G,q)$ be an instance of \MCI, where $G$ is a graph with a partition $V_1, \ldots, V_q$ of its
    vertex set into $q$ cliques, each of size $n$.
    For every $1 \le i < j \le q$, let $C^{i,j}$ denote the set of edges between $V_i$ and $V_j$ in $G$.
    We construct an instance $\mathcal{J}=(H,k,\col)$ of {\MCS} with
    $\col \colon V(H) \to \{\red,\blue\}$ and $k = 2q + 3$.
    Throughout the construction, whenever we say that we connect two vertices by a path
    with $t$ red (respectively blue) vertices, we mean a path whose $t$ internal
    vertices all receive color red (respectively blue); in particular, the distance
    between the two endpoints of such a path is $t+1$.
    All paths introduced below are internally vertex-disjoint.
    A local view of the construction is depicted in \cref{fig:improved-whard:reduction-v2}.

    \paragraph{Construction.}
    We use a small forcing structure called a \emph{marked vertex} several times.

    \minipara{Marked vertices.} When we say that we introduce a \emph{marked} vertex $z$,
    we give $z$ an arbitrary color $\gamma\in\{\red,\blue\}$, attach $k+1$ first-level leaves of the other color,
    and then attach $k+1$ second-level leaves of color $\gamma$ to each first-level leaf.
    The vertices in this two-level leaf structure are unnamed and are used only for forcing.

    The role of a marked vertex is the following. In any consistent subset of size at
    most $k$, neither $z$ nor any vertex in its attached two-level leaf structure can
    be selected. Moreover, if $d$ is the minimum distance from $z$ to a selected
    vertex, then the selected vertices at distance $d$ from $z$ must contain both
    colors. We prove this formally in the reverse direction; in the forward direction
    we only use the immediate converse, namely that such a pair of closest selected
    vertices satisfies $z$ and all vertices in its attached leaf structure.

    \minipara{Global gadget.}
    Introduce a marked vertex $m_g$, a blue vertex $b_g$ and a red vertex $r_g$.
    Add the edges $\{m_g,b_g\}$ and $\{m_g,r_g\}$.

    \minipara{Selection gadgets.}
    For every $i\in [q]$ we introduce a selection gadget with two components.

    For the blue component,
    introduce marked vertices $m^{i}_{\mathrm b}$ and $m^{i}_{\mathrm r}$ and two blue terminal vertices $\beta^i_1,\beta^i_2$.
    Connect $m^{i}_{\mathrm b}$ to $r_g$ by a path with $n$ red internal vertices.
    For every $s\in[n]$, introduce a blue vertex $v^i_s$ and connect it to $m^{i}_{\mathrm b}$ by a path with $n$ blue internal vertices and to $m^{i}_{\mathrm r}$ by a path with $2n+5$ blue internal vertices. Finally, connect each $v^i_s$ to $\beta^i_1$
    by a path with $s$ blue internal vertices and to $\beta^i_2$ by a path with
    $n-s+1$ blue internal vertices.

    For the red component, introduce two red terminal vertices $\gamma^i_1$ and $\gamma^i_2$. For every $s\in[n]$, introduce a red vertex $u^i_s$ and connect it to $m^{i}_{\mathrm r}$ by a path with $2n+5$ red internal vertices, to $\gamma^i_1$ by a path with $s$ red internal vertices, and to $\gamma^i_2$ by a path with $n-s+1$ red internal vertices.

    \minipara{Forcing gadget.}
    Introduce a marked vertex $m^\star$ and a blue vertex $v^\star$. Connect
    $m^\star$ to $r_g$ by a path with $n$ red internal vertices, and to
    $v^\star$ by a path with $n$ blue internal vertices.

    \minipara{Conflict gadgets.}
    For every pair $1\le i<j\le q$ and every $(s,t)\in C^{i,j}$ we introduce a marked vertex $w^{i,j}_{s,t}$, called a \emph{conflict vertex}, together with a blue vertex $b^{i,j}_{s,t}$ and the following paths.
    \begin{itemize}
        \item Blue paths from $b^{i,j}_{s,t}$ to $w^{i,j}_{s,t}$, $m^\star$, and $v^\star$ with $n$, $2n+5$, and $2n+3$ blue internal vertices, respectively.
        \item Blue paths from $b^{i,j}_{s,t}$ to
            $\beta^{i}_2,\beta^{i}_1,\beta^{j}_2,\beta^{j}_1$
            with, respectively, $n+s+2$, $2n-s+3$, $n+t+2$, $2n-t+3$ blue internal
            vertices.
        \item Red paths from $w^{i,j}_{s,t}$ to
            $\gamma^{i}_2,\gamma^{i}_1,\gamma^{j}_2,\gamma^{j}_1$
            with the same respective numbers of red internal vertices, each increased by $n+1$.
    \end{itemize}

    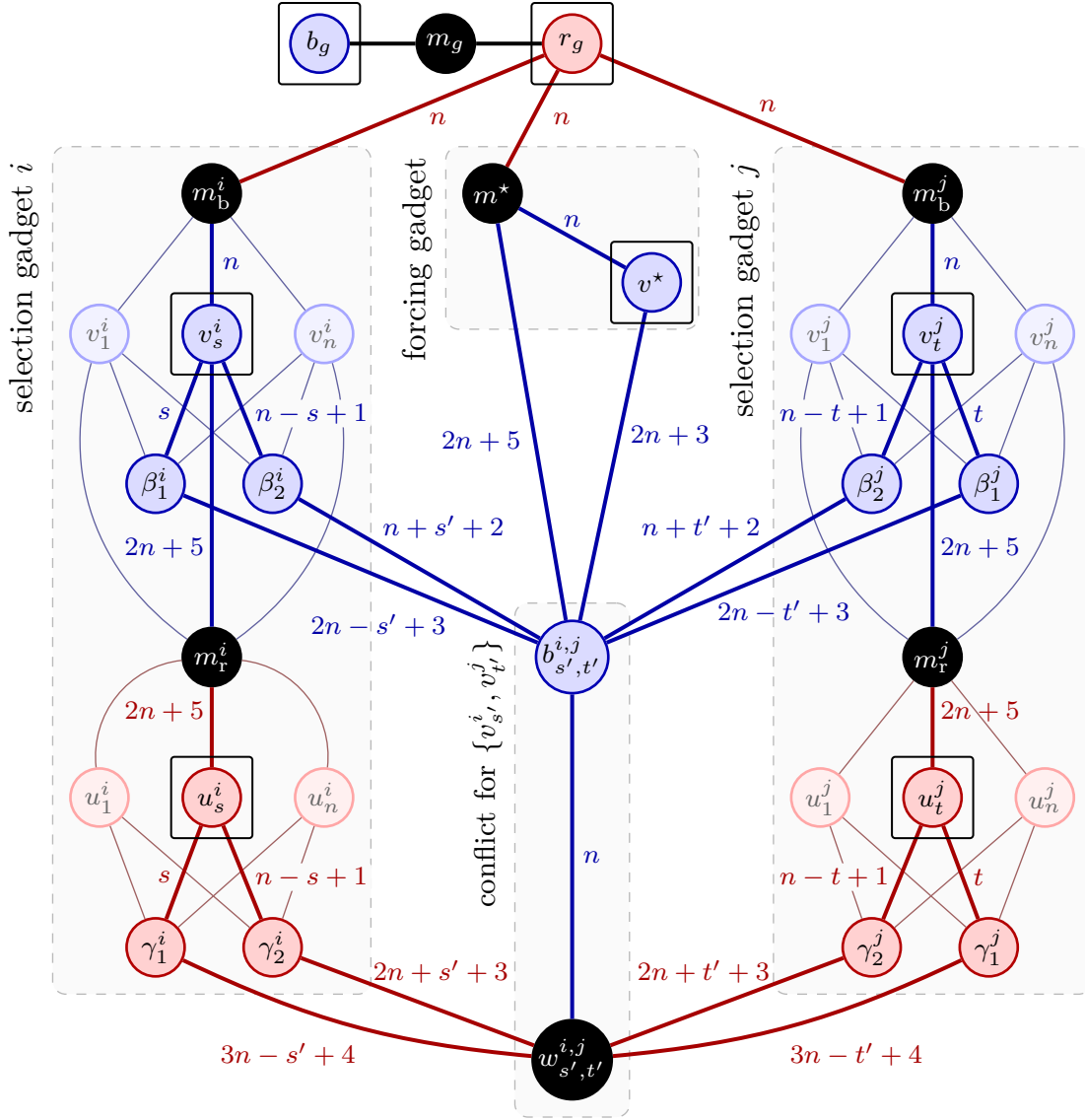
\begin{figure}[ht!]
        \centering
        \resizebox{.9\linewidth}{!}{\input{graphics/whardness_linearFVS}}
        \caption{A local view of the construction in
        \cref{thm:whard:unweighted-c-pw-fvs-k} for a single conflicting
        pair $(s,t)\in C^{i,j}$ with $1\le i<j\le q$. Individual edges are unlabeled, black and bold. Paths are drawn as single
        colored edges labeled by the number of internal vertices where highlighted.
        Vertices with edges facing outside this local view are $r_g$ (to other selection gadgets) as well as the $\beta$ and $\gamma$ vertices (to other conflict gadgets). Squares indicate selected vertices in a hypothetical solution.}
        \label{fig:improved-whard:reduction-v2}
    \end{figure}

    This completes the construction of $H$. Clearly it is computable in polynomial time.

    \paragraph{Bounds on parameters.}
    Let $F=\{m^\star\}\cup \setdef{m^i_{\mathrm b}, \beta^i_1, \beta^i_2, \gamma^i_1, \gamma^i_2}{i\in [q]}$.
    Then, every connected component in $H-F$ is a subdivided star (the global gadget and the forcing gadget), two subdivided stars connected by a path (the conflict gadgets), or $2n$ subdivided stars all connected to a central vertex $m^i_{\mathrm r}$ (the selection gadgets).
    Thereby, $F$ is a feedback vertex set witnessing $\fvs \le |F| = 5q+1$ for $H$ and also implies the existence of a path decomposition of width $\bO(|F|)=\bO(q)$.

    The treedepth of $H$ is at most $|F|$ plus the largest treedepth among the connected components in $H-F$. Recall from above the shape of these components. For a subdivided star, removing the center contributes one to the count, leaving only paths of length $\bO(n)$. For two connected subdivided stars, removing both centers has the same effect. For the selection gadget, removing the central vertex $m^i_{\mathrm r}$ gives a set of pairwise disjoint and disconnected subdivided stars; we apply the above argument to the one with largest treedepth. Thus, after eliminating at most $|F|+2$ vertices we end up with a path of length $\bO(n)$. It has treedepth $\bO(\log n)$ as seen by recursively removing the central vertex and recursing in the larger half.
    Thus the treedepth of $H$ is bounded by $\bO(|F|+\log n) = \bO(q+\log n)$.

    \paragraph{Correctness: \MCIshort $\to$ \MCSshort.}
    Let $\sigma\colon[q]\to[n]$ be such that
    $C=\setdef{v^i_{\sigma(i)}}{i\in[q]}$ is a multi-colored independent set in $G$. Let $\sigma_i = \sigma(i)$ for brevity.
    We define
    \[
        S = \{b_g,r_g,v^\star\}
        \cup \setdef{v^i_{\sigma_i}}{i\in[q]}
        \cup \setdef{u^i_{\sigma_i}}{i\in[q]}.
    \]
    Then $|S|=3+q+q=k$. We show that $S$ is a consistent subset of $H$.

        \minipara{Marked vertices.}
        We first note how marked vertices are satisfied by $S$. Let $z$ be a marked
        vertex, and let $d$ be the minimum distance from $z$ to a vertex in $S$. Suppose
        that $S$ contains a red and a blue vertex at distance $d$ from $z$. Since no
        vertex in the two-level leaf structure at $z$ belongs to $S$, every shortest
        path from a first-level or second-level unnamed leaf to a selected vertex passes
        through $z$. Hence $z$ is satisfied by the closest selected vertex of its own
        color, every first-level leaf is satisfied by the closest selected vertex of its
        color, and every second-level leaf is satisfied in the same way.

        The marked vertex $m_g$ has the selected blue vertex $b_g$ and the selected red vertex $r_g$ at distance one.
        The marked vertex $m^\star$ has the selected blue vertex $v^\star$ and the selected red vertex $r_g$ at distance $n+1$ and no closer selected vertices.

        The marked vertex $m^{i}_{\mathrm b}$ for each $i\in [q]$ has distance
        $n+1$
        to the selected red vertex $r_g$ and the selected blue vertex $v^i_{\sigma_i}$.
        The marked vertex $m^{i}_{\mathrm r}$ has distance $2n+6$ to the selected blue and red vertices $v^i_{\sigma_i}$ and $u^i_{\sigma_i}$.
        Again, for both cases there are no closer selected vertices (all closer vertices are unnamed internal vertices on the attached paths).

        The only remaining marked vertices are conflict vertices. Consider any conflict vertex $w=w^{i,j}_{s,t}$, and let $b=b^{i,j}_{s,t}$ be the other named vertex in the conflict gadget.
        The distance from $w$ to $v^i_{\sigma_i}$ via $b$ and $\beta^i_2$ is
        \[
            (n+1) + (n+s+3) + (n-\sigma_i+2) = 3n+6 - (\sigma_i - s),
        \]
        and via $\beta^i_1$ it is
        \[
            (n+1) + (2n-s+4) + (\sigma_i+1) = 3n+6 + (\sigma_i - s).
        \]
        The shorter of the two equals $3n+6-|\sigma_i-s|$. The same calculation on the
        $j$-side gives $\dist(w,v^j_{\sigma_j}) = 3n+6-|\sigma_j-t|$, where we abbreviate $\dist_H(\cdot)$ to $\dist(\cdot)$ for the remainder of the proof.
        On the red side, $u^i_{\sigma_i}$ is reached from $w$ directly via $\gamma^i_x$
        (without the intermediate $b$ hop), so
        \[
            \dist(w,u^i_{\sigma_i}) = \min\bigl((2n+s+4)+(n-\sigma_i+2),\ (3n-s+5)+(\sigma_i+1)\bigr)
            = 3n+6 - |\sigma_i-s|,
        \]
        and symmetrically $\dist(w,u^j_{\sigma_j}) = 3n+6 - |\sigma_j-t|$.
        Since $C$ is an independent set, $\{v^i_{\sigma_i},v^j_{\sigma_j}\}\notin E(G)$, hence
        $(\sigma_i,\sigma_j)\notin C^{i,j}$, hence $(s,t)\ne(\sigma_i,\sigma_j)$.
        Therefore at least one of $|\sigma_i-s|, |\sigma_j-t|$ is positive.
        Without loss of generality, let $|\sigma_i-s|$ be the larger one and
        \[
            d^\star = 3n+6 - |\sigma_i-s| \le 3n+5.
        \]
        Then $w$ sees both a blue and a red vertex ($v^i_{\sigma_i}$ and $u^i_{\sigma_i}$) at distance $d^\star$.

        Moreover, no other selected vertex is closer to $w$:
        The forcing gadget is reached through $b$:
        $\dist(w,v^\star) = \dist(w,b)+\dist(b,v^\star) = (n+1)+(2n+4) = 3n+5$
        and
        $\dist(w,m^\star) = \dist(w,b)+\dist(b,m^\star) = (n+1)+(2n+6) = 3n+7$.
        Reaching $r_g$ (and via it $b_g$) without using the forcing gadget requires passing through some $v^i_s$ and then $m^i_{\mathrm b}$,
        thus $\dist(v^i_s,w) < \dist(r_g,w) < \dist(b_g,w)$;
        via $m^\star$ and $b$ we get $\dist(w, r_g) = 4n+8 > d^\star$.
        Selected vertices in other selection gadgets are also not closer:
        going via another $b'$ and some $\beta^j_y$ in another conflict gives, for some $\beta^i_x$, distance at least
        \[\dist(w,b)+\dist(b,\beta^i_x)+\dist(\beta^i_x,b')+\dist(b',\beta^j_y),\]
        each summand at least $n+1$, so a total of at least $4(n+1) > 3n+5 \ge d^\star$.
        Going via another conflict vertex $w'$ similarly gives, for some $\gamma^i_x$ in this gadget and $\gamma^j_y$ in another, distance at least
        \[\dist(w,\gamma^i_x)+\dist(\gamma^i_x,w'),\]
        with each summand exceeding $2n$, hence a total exceeding $4n \ge d^\star$.

        Thus all marked vertices and their attached unnamed vertices are satisfied by the observation above.

        \minipara{Non-marked vertices.}
        Observe that after removing all marked vertices and their attached unnamed vertices, every connected component is monochromatic.
        Consider any vertex $v$ in such a component. If a shortest path from $v$ to $S$ does not pass through a marked vertex, then it leads to a selected vertex of the same color, so $v$ is satisfied. Otherwise, since every marked vertex sees both colors at its minimum distance to $S$, the same holds for $v$.

        Therefore every vertex of $H$ is satisfied, and $\mathcal{J}$ is a yes-instance of \MCS.

    \paragraph{Correctness: \MCSshort $\to$ \MCIshort.}
    Let $S \subseteq V(H)$ be a consistent subset of $H$ of size at most $k$. Before describing how to derive a solution to \MCI, we first establish some structural properties of $S$.

    Let $M$ be the set of marked vertices and their attached unnamed vertices in $H$.

    \begin{claim}
        It holds that $S \cap M = \varnothing$.
        Furthermore, for every marked vertex $z$, the vertices of $S$ at minimum distance from $z$
        include both a red vertex and a blue vertex.
    \end{claim}

    \begin{claimproof}
        Let $z$ be a marked vertex and let $\gamma \in \{\red,\blue\}$ be its color;
        write $\bar\gamma$ for the other color. Let $L_{\bar\gamma}$ be the set of
        the $k+1$ first-level unnamed leaves attached to $z$, and for
        $x\in L_{\bar\gamma}$ let $L_\gamma(x)$ be the set of the $k+1$
        second-level unnamed leaves attached to $x$.

        We first prove that none of these vertices is selected. If $z\in S$, then
        some $x\in L_{\bar\gamma}$ is not selected. Its closest selected vertex is then $z$, of color $\gamma$, at distance one; since all neighbors of $x$ have color $\gamma$, no selected
        vertex of color $\bar\gamma$ is at distance at most one, contradicting consistency. Hence $z\notin S$.

        If some first-level leaf $x\in L_{\bar\gamma}$ belongs to $S$, then some
        $y\in L_\gamma(x)$ is not selected. The vertex $y$ sees the selected vertex
        $x$ of the opposite color at distance one, and since $x$ is its only neighbor, no selected vertex of color
        $\gamma$ is at distance at most one---a contradiction. Thus no first-level
        leaf is selected.

        Finally, if some second-level leaf $y\in L_\gamma(x)$ belongs to $S$, then
        $x$ sees the selected vertex $y$ of the opposite color at distance one, but
        no selected vertex of color $\bar\gamma$ at distance at most one. This is
        impossible, so no second-level leaf is selected either. Therefore
        none of the vertices in the marked structure rooted at $z$ belongs to $S$.
        Since $z$ was arbitrary, $S\cap M=\varnothing$.

        Since no vertex in the two-level leaf structure at $z$ is selected, every
        shortest path from a first-level or second-level leaf to a selected vertex
        passes through $z$. Hence the closest selected vertices to any first-level
        leaf are exactly the vertices of $S$ at distance $\dist(z,S)$ from $z$. Since $z$ has color $\gamma$ and first-level leaves have color $\bar\gamma$,
        consistency forces these vertices to include both a red and a blue one.
    \end{claimproof}

    \begin{claim}
        Ignoring $b_g, r_g$ and vertices in marked structures,
        every blue vertex is at distance at least $n+2$ and at most $4n+7$ from $r_g$ and
        every red vertex is at distance at most $n$ or more than $4n+9$ from $r_g$.
    \end{claim}
    \begin{claimproof}
        Inspecting the construction, the relevant vertices at distance at most $n$ are all red, and the vertices at distance $n+1$ are all marked ($m^i_{\mathrm b}$ for $i\in [q]$ and $m^\star$).
        Fix $i\in [q]$.
        Using the attached paths via $m^i_{\mathrm b}$, we have $\dist(r_g, v_s^i) = (n+1)+(n+1)= 2n+2$ for every $s\in [n]$, which proves the claim for these blue vertices and the internal vertices on the traversed paths.
        The blue paths from $v_s^i$ to $m^i_{\mathrm r}$ have $2n+5$ internal vertices each, so the furthest such vertex is at distance $(2n+2)+(2n+5)=4n+7$ from $r_g$.

        Any internal vertex on a path between some $v_s^i$ and $\beta_1^i$ (or $\beta_2^i$) is at most $n$ steps from $v_s^i$ and hence at distance at most $3n+2 \le 4n+7$ from $r_g$.
        The two $\beta$ vertices have distance $(2n+2)+2=2n+4$ from $r_g$, via $v_1^i$ and $v_n^i$ respectively.

        For a conflict gadget with vertices $b = b^{i,j}_{s,t}$ and $w = w^{i,j}_{s,t}$, the vertex $b$ has distance at most $(n+1)+(2n+6)=3n+7$ from $r_g$ via $m^\star$ and the attached path (all internal vertices on that path are even closer).
        The $n$ blue internal vertices on the $b$--$w$ path therefore have distance at most $4n+7$ from $r_g$.

        Since each $\beta$ vertex has distance $2n+4$ from $r_g$, any of the at most $2n+2$ internal vertices on the $b$--$\beta_1^i$ (or $b$--$\beta_2^i$) path has distance at most $(2n+4)+(2n+2) = 4n+6$ from $r_g$ via the corresponding $\beta$ vertex.

        The vertex $v^\star$ has distance $2n+2$ from $r_g$ via $m^\star$. Every internal vertex on the two paths attached to $v^\star$ has distance at most $2n+3$from $v^\star$, hence at most $(2n+2)+(2n+3)=4n+5$ from $r_g$. Every one of the $2n+5$ internal vertices on the path connecting $m^\star$ to some $b$ has distance at most $(2n+5)+\dist(m^\star,r_g) = 3n+6$ from $r_g$.
        This completes the bound for the blue vertices.

        Any non-marked red vertex at distance more than $n$ from $r_g$ can only be reached via a marked $m^i_{\mathrm r}$ or a marked conflict vertex $w$. The bounds above place these marked vertices at distance $4n+8$, so any remaining red vertex is at distance at least $4n+9$.
        \[
        \begin{array}{r@{\,}c@{\,}l|c}
            \multicolumn{3}{c|}{\text{distance interval from } r_g}
                & \text{color of vertices in } V(H)\setminus M \\ \hline
            1            & ,\ldots, & n         & \red \\
            n+2        & ,\ldots, & 4n+7     & \blue\\
            4n+9    & ,\ldots, &  \infty & \red
        \end{array}
        \]
        The omitted layers contain only vertices of $M$.
    \end{claimproof}

    \begin{claim}
        It holds that $r_g,b_g \in S$.
    \end{claim}

    \begin{claimproof}
        By the first claim, the closest selected vertices to $m_g$ contain both
        colors. Let $d = \dist(m_g,S)$. Since $m_g\notin S$, $d\ge 1$, and we show that $d=1$.

        Suppose for a contradiction that $d>1$. Then neither $b_g$ nor $r_g$ belongs to
        $S$. As $b_g$ is a leaf outside the marked structure at $m_g$, every closest selected vertex to $m_g$ is reached through $r_g$.

        Consider the distance layers from $r_g$ among the vertices of $V(H)\setminus M$, ignoring the leaf $b_g$. The previous claim shows that these layers are monochromatic (the omitted ones contain only marked vertices). Hence no distance layer larger than one contains selected vertices of both colors, contradicting the marked-vertex claim for $m_g$.
        Consequently $d=1$, and the only selectable vertices at distance one from $m_g$ are $b_g$ and $r_g$. Since both colors must appear in this layer, $b_g, r_g \in S$.
    \end{claimproof}

    \begin{claim}
        $S = \{r_g, b_g, s^\star, v^1_{s_1}, u^1_{t_1}, \ldots, v^k_{s_q}, u^k_{t_q}\}$, where $s_i, t_i \in [n]$ for $i\in [q]$ and $s^\star$ lies in the forcing gadget or on a path connecting $m^\star$ to some conflict gadget.
    \end{claim}

    \begin{claimproof}
        The previous claim gives $\{r_g, b_g\} \subseteq S$.
        We also know that every marked vertex has both a selected blue and a selected red vertex at its minimum distance to $S$.
        Fix $i \in [q]$. The marked vertex $m^i_{\mathrm b}$ has a closest selected red vertex at distance at most $n+1$, namely $r_g \in S$. Every blue vertex within that distance from $m^i_{\mathrm b}$ lies in the blue component of selection gadget $i$, so at least one blue vertex is selected from each selection gadget.
        Such a vertex must be $v^i_{s_i}$ for some $s_i \in [n]$ (if $r_g$ is the closest selected red vertex to $m^i_{\mathrm b}$) or an internal vertex on one of the $n$-vertex paths connecting some $v^i_{s_i}$ to $m^i_{\mathrm b}$ (if a selected red vertex is closer).

        Similarly, at least one vertex $s^\star$ must be selected for the forcing gadget.
        Now consider $m^i_{\mathrm r}$: it has distance $2n+6$ to every $v^i_{s}$, hence distance at most $(2n+6)+n$ to some selected blue vertex inside the gadget. Every red vertex within that distance from $m^i_{\mathrm r}$ lies either in the gadget or on a path between the gadget and a conflict vertex; moreover, no red vertex can serve $m^i_{\mathrm r}$ for two distinct values of $i$.

        Consequently at least two vertices are selected from each selection gadget, plus two from the global gadget and one from the forcing gadget, exhausting the budget $k=2q+3$. Thus no further vertex is selected. In particular, $r_g$ is the closest selected red vertex to $m^i_{\mathrm b}$, at distance $n+1$, so the selected blue vertex from selection gadget $i$ must be some $v^i_{s_i}$. This vertex is at distance exactly $2n+6$ from $m^i_{\mathrm r}$, so the selected red vertex in the gadget lies at the same distance, namely some $u^i_{t_i}$.
    \end{claimproof}

    \begin{claim}
        Let $u^1_{t_1}, \ldots, u^q_{t_q}$ be the selected red vertices from the selection gadgets. Then $v^1_{t_1}, \ldots, v^q_{t_q}$ form a multi-colored independent set in $G$.
    \end{claim}

    \begin{claimproof}
        We first show that $s^\star = v^\star$. By the previous claim $s^\star$ lies in the forcing gadget or on the $m^\star$--$b^{i,j}_{s,t}$ path of some conflict gadget. In every case $s^\star$ is blue and $r_g$ is the closest selected red vertex to $m^\star$ at distance $n+1$, so the marked-vertex condition for $m^\star$ forces $\dist(s^\star, m^\star)=n+1$. If $s^\star$ lies on the $m^\star$--$v^\star$ path this immediately gives $s^\star=v^\star$.

        Suppose instead that $s^\star$ lies on the $m^\star$--$b$ path for some conflict gadget with conflict vertex $w=w^{i,j}_{s,t}$ and blue hub $b=b^{i,j}_{s,t}$. Since the path has $2n+5$ internal vertices and $s^\star$ is the $(n+1)$-th, $\dist(s^\star, b) = (2n+6)-(n+1) = n+5$, and hence
        \[
            \dist(w,S) \le \dist(s^\star,b)+\dist(b,w) = (n+5)+(n+1) = 2n+6.
        \]
        The marked-vertex condition for $w$ then requires a selected red vertex within distance $2n+6$. But $r_g$ is much further (the shortest route via $m^\star$ and $b$ has length $4n+8$), and for any $a\in[n]$,
        \[
            \dist(w, u^i_a) = 3n+6 - |a-s| \ge 2n+7,
        \]
        with the symmetric bound on the $j$-side; no other selected red vertex is closer. This contradicts $\dist(w,S)\le 2n+6$, so $s^\star=v^\star$.

        We now show that no edge of $G$ joins two of the vertices $v^1_{t_1},\ldots,v^q_{t_q}$. Suppose for a contradiction that $\{v^i_{t_i}, v^j_{t_j}\}\in E(G)$ for some $1\le i<j\le q$. Setting $s=t_i$ and $t=t_j$, the conflict vertex $w=w^{i,j}_{s,t}$ exists in $H$. Since $v^\star\in S$ and $\dist(w,v^\star)=3n+5$, we have $\dist(w,S)\le 3n+5$, so the marked-vertex condition for $w$ requires a selected red vertex within distance $3n+5$. The vertex $r_g$ and any $u^{i'}_{t_{i'}}$ with $i'\notin\{i,j\}$ are at distance more than $4n$ from $w$, so the only candidates are $u^i_{t_i}$ and $u^j_{t_j}$. However, going via $\gamma^i_1$ or $\gamma^i_2$,
        \[
            \dist(w,u^i_{t_i}) = \min\bigl((3n-s+5)+(s+1),\ (2n+s+4)+(n-s+2)\bigr) = 3n+6,
        \]
        and symmetrically $\dist(w,u^j_{t_j})=3n+6$. Hence no selected red vertex is within distance $3n+5$ of $w$, a contradiction.
    \end{claimproof}
\end{proof}

%% file: graphics/whardness_linearFVS.tex
% Graphic created using the help of Claude Opus 4.7
% Illustration of the construction in \cref{thm:improved-whard:unweighted-c-pw-fvs-k}.
% Paths are compressed to single colored edges; labels indicate the number of
% internal vertices on the represented path.
%
% Drawing order: vertices, background boxes, edges (unlabeled), labels,
% selection highlights. Labels are emitted last so they sit above every edge
% in their vicinity.
\begin{tikzpicture}[
    x=1cm, y=1cm,
    font=\scriptsize,
    vertex/.style={circle, minimum size=6.2mm, inner sep=0pt, thick},
    marked/.style={vertex, draw=black, fill=black, text=white},
    redv/.style={vertex, draw=red!70!black, fill=red!18, text=black},
    bluev/.style={vertex, draw=blue!70!black, fill=blue!14, text=black},
    faintred/.style={redv, draw=red!35, fill=red!6, text=black!60},
    faintblue/.style={bluev, draw=blue!35, fill=blue!5, text=black!60},
    ordinary/.style={very thick},
    rpath/.style={red!65!black, very thick},
    bpath/.style={blue!65!black, very thick},
    mutedrpath/.style={red!25!gray},
    mutedbpath/.style={blue!25!gray},
    len/.style={inner sep=1.2pt, font=\scriptsize, rounded corners=1pt},
    rlen/.style={len, text=red!65!black},
    blen/.style={len, text=blue!65!black},
    box/.style={draw=gray!55, rounded corners=4pt, dashed, inner sep=5pt},
    selected/.style={draw=black, rounded corners=1pt, line width=0.6pt, inner sep=1.1mm},
    >=stealth,
]

% =========================================================================
% Vertices
% =========================================================================

% Global gadget
\node[marked] (mg) at (-1.35, 1.05) {$m_g$};
\node[bluev]  (bg) at (-2.7,  1.05) {$b_g$};
\node[redv]   (rg) at (0,     1.05) {$r_g$};

% Selection gadget i (blue side)
\node[marked]    (mib) at (-3.85, -0.55) {$m^{i}_{\mathrm b}$};
\node[faintblue] (vin) at (-2.65, -2.05) {$v^{i}_n$};
\node[bluev]     (vis) at (-3.85, -2.05) {$v^{i}_s$};
\node[faintblue] (vi1) at (-5.05, -2.05) {$v^{i}_1$};
\node[bluev]     (bi2) at (-3.20, -3.65) {$\beta^{i}_2$};
\node[bluev]     (bi1) at (-4.45, -3.65) {$\beta^{i}_1$};

% Selection gadget j (blue side)
\node[marked]    (mjb) at (3.85,  -0.55) {$m^{j}_{\mathrm b}$};
\node[faintblue] (vjn) at (5.05,  -2.05) {$v^{j}_n$};
\node[bluev]     (vjt) at (3.85,  -2.05) {$v^{j}_{t}$};
\node[faintblue] (vj1) at (2.65,  -2.05) {$v^{j}_1$};
\node[bluev]     (bj2) at (3.20,  -3.65) {$\beta^{j}_2$};
\node[bluev]     (bj1) at (4.45,  -3.65) {$\beta^{j}_1$};

% Selection gadget i (red side)
\node[marked]    (mir) at (-3.85, -5.5)  {$m^{i}_{\mathrm r}$};
\node[faintred]  (uin) at (-2.65, -7)    {$u^{i}_n$};
\node[redv]      (uis) at (-3.85, -7)    {$u^{i}_s$};
\node[faintred]  (ui1) at (-5.05, -7)    {$u^{i}_1$};
\node[redv]      (gi2) at (-3.20, -8.6)  {$\gamma^{i}_2$};
\node[redv]      (gi1) at (-4.45, -8.6)  {$\gamma^{i}_1$};

% Selection gadget j (red side)
\node[marked]    (mjr) at (3.85,  -5.5)  {$m^{j}_{\mathrm r}$};
\node[faintred]  (ujn) at (5.05,  -7)    {$u^{j}_n$};
\node[redv]      (ujt) at (3.85,  -7)    {$u^{j}_{t}$};
\node[faintred]  (uj1) at (2.65,  -7)    {$u^{j}_1$};
\node[redv]      (gj2) at (3.20,  -8.6)  {$\gamma^{j}_2$};
\node[redv]      (gj1) at (4.45,  -8.6)  {$\gamma^{j}_1$};

% Forcing gadget
\node[marked] (mstar) at (-0.85, -0.55) {$m^{\star}$};
\node[bluev]  (vstar) at ( 0.85, -1.5)  {$v^{\star}$};

% Conflict gadget
\node[bluev]  (bij) at (0, -5.5) {$b^{i,j}_{s',t'}$};
\node[marked] (w)   at (0, -9.8) {$w^{i,j}_{s',t'}$};

% =========================================================================
% Background boxes (on background layer, drawn behind everything)
% =========================================================================
\begin{scope}[on background layer]
    \node[box, fill=gray!4,
          fit=(mib)(vin)(vis)(vi1)(bi2)(bi1)(mir)(uin)(uis)(ui1)(gi2)(gi1),
          label={[font=\small, xshift=-3mm, yshift=4.5cm, rotate=90]left:selection gadget $i$}] {};
    \node[box, fill=gray!4,
          fit=(mjb)(vjn)(vjt)(vj1)(bj2)(bj1)(mjr)(ujn)(ujt)(uj1)(gj2)(gj1),
          label={[font=\small, xshift=-3mm, yshift=4.5cm, rotate=90]left:selection gadget $j$}] {};
    \node[box, fill=gray!4,
          fit=(mstar)(vstar),
          label={[font=\small, xshift=-3mm, yshift=1cm, rotate=90]left:forcing gadget}] {};
    \node[box, fill=gray!4,
          fit=(bij)(w),
          label={[font=\footnotesize, xshift=-3mm, yshift=2.5cm, rotate=90]left:conflict for $\{v^i_{s'},v^j_{t'}\}$}] {};
\end{scope}

% =========================================================================
% Edges (unlabeled). Labels are placed in a later block so they sit on top.
% =========================================================================

% Global gadget
\draw[ordinary] (mg) -- (rg);
\draw[ordinary] (mg) -- (bg);

% Selection gadget i: blue side
\draw[rpath]      (rg)  -- (mib);
\draw[mutedbpath] (mib) -- (vin);
\draw[bpath]      (mib) -- (vis);
\draw[mutedbpath] (mib) -- (vi1);
\draw[mutedbpath] (vin) -- (bi2);
\draw[mutedbpath] (vin) -- (bi1);
\draw[mutedbpath] (vi1) -- (bi2);
\draw[mutedbpath] (vi1) -- (bi1);
\draw[bpath]      (vis) -- (bi2);
\draw[bpath]      (vis) -- (bi1);

% Selection gadget j: blue side
\draw[rpath]      (rg)  -- (mjb);
\draw[mutedbpath] (mjb) -- (vjn);
\draw[bpath]      (mjb) -- (vjt);
\draw[mutedbpath] (mjb) -- (vj1);
\draw[mutedbpath] (vjn) -- (bj2);
\draw[mutedbpath] (vjn) -- (bj1);
\draw[mutedbpath] (vj1) -- (bj2);
\draw[mutedbpath] (vj1) -- (bj1);
\draw[bpath]      (vjt) -- (bj2);
\draw[bpath]      (vjt) -- (bj1);

% Selection gadget i: red side and m^i_r connections
\draw[mutedbpath] (mir) to[bend right=35] (vin);
\draw[bpath]      (mir) -- (vis);
\draw[mutedbpath] (mir) to[bend left=35]  (vi1);
\draw[mutedrpath] (mir) to[bend left=45] (uin);
\draw[rpath]      (mir) -- (uis);
\draw[mutedrpath] (mir) to[bend right=45] (ui1);
\draw[mutedrpath] (uin) -- (gi2);
\draw[mutedrpath] (uin) -- (gi1);
\draw[mutedrpath] (ui1) -- (gi2);
\draw[mutedrpath] (ui1) -- (gi1);
\draw[rpath]      (uis) -- (gi2);
\draw[rpath]      (uis) -- (gi1);

% Selection gadget j: red side and m^j_r connections
\draw[mutedbpath] (mjr) to[bend left=35]  (vj1);
\draw[bpath]      (mjr) -- (vjt);
\draw[mutedbpath] (mjr) to[bend right=35] (vjn);
\draw[mutedrpath] (mjr) -- (ujn);
\draw[rpath]      (mjr) -- (ujt);
\draw[mutedrpath] (mjr) -- (uj1);
\draw[mutedrpath] (ujn) -- (gj2);
\draw[mutedrpath] (ujn) -- (gj1);
\draw[mutedrpath] (uj1) -- (gj2);
\draw[mutedrpath] (uj1) -- (gj1);
\draw[rpath]      (ujt) -- (gj2);
\draw[rpath]      (ujt) -- (gj1);

% Forcing gadget
\draw[rpath] (rg)    -- (mstar);
\draw[bpath] (mstar) -- (vstar);

% Conflict gadget
\draw[bpath] (bij) -- (bi2);
\draw[bpath] (bij) -- (bi1);
\draw[bpath] (bij) -- (bj2);
\draw[bpath] (bij) -- (bj1);
\draw[bpath] (bij) -- (mstar);
\draw[bpath] (vstar) -- (bij);
\draw[bpath] (bij) -- (w);
\draw[rpath] (w) --                 (gi2);
\draw[rpath] (w) to[bend left=10]   (gi1);
\draw[rpath] (w) --                 (gj2);
\draw[rpath] (w) to[bend right=10]  (gj1);

% =========================================================================
% Edge labels (placed after edges so they sit above every nearby edge)
% =========================================================================

% Global - selection / forcing connecting edges from r_g
\path (rg) -- node[rlen, pos=.5, right=10pt]    {$n$} (mib);
\path (rg) -- node[rlen, pos=.5, above right=1pt]{$n$} (mjb);
\path (rg) -- node[rlen, pos=.5, right=5pt]     {$n$} (mstar);

% Selection gadget i: blue side
\path (mib) -- node[blen, pos=.5, right=2pt] {$n$}       (vis);
\path (vis) -- node[blen, pos=.55, right=2pt, fill=gray!4] {$n-s+1$}   (bi2);
\path (vis) -- node[blen, pos=.55, left=2pt]  {$s$}       (bi1);

% Selection gadget j: blue side
\path (mjb) -- node[blen, pos=.5, right=2pt] {$n$}       (vjt);
\path (vjt) -- node[blen, pos=.55, left=2pt, fill=gray!4]  {$n-t+1$}   (bj2);
\path (vjt) -- node[blen, pos=.55, right=2pt] {$t$}       (bj1);

% Selection gadget i: red side and m^i_r
\path (mir) -- node[blen, pos=.3, left=1pt]  {$2n+5$}    (vis);
\path (mir) -- node[rlen, pos=.3, left=1pt]  {$2n+5$}    (uis);
\path (uis) -- node[rlen, pos=.55, right=2pt, fill=gray!4] {$n-s+1$}   (gi2);
\path (uis) -- node[rlen, pos=.55, left=2pt]  {$s$}       (gi1);

% Selection gadget j: red side and m^j_r
\path (mjr) -- node[blen, pos=.3, right=1pt] {$2n+5$}    (vjt);
\path (mjr) -- node[rlen, pos=.3, right=1pt] {$2n+5$}    (ujt);
\path (ujt) -- node[rlen, pos=.55, left=2pt, fill=gray!4]  {$n-t+1$}   (gj2);
\path (ujt) -- node[rlen, pos=.55, right=2pt] {$t$}       (gj1);

% Forcing gadget
\path (mstar) -- node[blen, pos=.5, above=2pt] {$n$}     (vstar);

% Conflict gadget: blue side
\path (bij) -- node[blen, pos=.4,  above=12pt] {$n+s'+2$}  (bi2);
\path (bij) -- node[blen, pos=.45, below=10pt] {$2n-s'+3$} (bi1);
\path (bij) -- node[blen, pos=.4,  above=12pt] {$n+t'+2$}  (bj2);
\path (bij) -- node[blen, pos=.5,  below=8pt]  {$2n-t'+3$} (bj1);
\path (bij) -- node[blen, pos=.45, left=3pt]   {$2n+5$}    (mstar);
\path (bij) -- node[blen, pos=.5,  right=2pt]  {$n$}       (w);
\path (vstar) -- node[blen, pos=.4, right=2pt] {$2n+3$}    (bij);

% Conflict gadget: red side
\path (w) --                node[rlen, pos=.4, above=7pt]  {$2n+s'+3$} (gi2);
\path (w) to[bend left=10]  node[rlen, pos=.7, below=10pt] {$3n-s'+4$} (gi1);
\path (w) --                node[rlen, pos=.4, above=7pt]  {$2n+t'+3$} (gj2);
\path (w) to[bend right=10] node[rlen, pos=.7, below=10pt] {$3n-t'+4$} (gj1);

% =========================================================================
% Highlight one coherent selection
% =========================================================================
\foreach \z in {rg,bg,vis,vjt,vstar,uis,ujt} {
    \node[selected, fit=(\z)] {};
}

\end{tikzpicture}

%% file: Sections/weighted-reduction.tex
For the weighted setting, we show that even replacing treewidth with the more restrictive vertex cover number cannot improve the running time.

\begin{theorem}\label{thm:whard:weighted-vc}
    {\WMCS} cannot be solved in time $f(k+\vc)\cdot n^{o(k+\vc)}$ under the \textup{ETH},
    even for $c=2$ and unary-encoded edge-weights, for any computable function $f$.
\end{theorem}

\begin{proof}
    We prove by reduction from \MCI.
    Let $(G,q)$ be an instance of \MCI, where $G$ is a graph with a partition $V_1, \ldots, V_q$ of its
    vertex set into $q$ cliques, where $V_i = \{v^i_1, \ldots, v^i_n\}$.
    We may assume that $n\ge 3$.
    We construct an equivalent {\WMCS} instance $(G', w, k, \col)$ with $k \coloneq q+1$ and two colors as follows.

    First, introduce a vertex $g$ to $G'$.
    Then, for every $i\in[q]$ we construct a group $W_i$ of vertices in $G'$,
    consisting of three vertices $a_i, b^1_i, b^2_i$, as well as a \emph{selector vertex} $w^i_j$ for every vertex $v^i_j \in V_i$.
    We connect $a_i$ to $g$ by an edge of weight $n^2$,
    and every vertex $w^i_j$ to $a_i$, $b^1_i$, and $b^2_i$ by edges of weight $n^2$, $j$, and $n-j+1$, respectively.

    Lastly, for every edge $e = \{v^i_j, v^{i'}_{j'}\} \in E(G)$, where $v^i_j \in V_i$ and $v^{i'}_{j'} \in V_{i'}$ with $i \neq i'$,
    introduce an \emph{edge vertex} $x_e$ to $G'$ and add the following incident edges:
    \begin{itemize}
        \item $\{x_e, g\}$ of weight $2n$,
        \item $\{x_e, b^1_i\}$ of weight $2n-j+1$ and $\{x_e, b^1_{i'}\}$ of weight $2n-j'+1$,
        \item $\{x_e, b^2_i\}$ of weight $n+j$ and $\{x_e, b^2_{i'}\}$ of weight $n+j'$.
    \end{itemize}

    Let $\col$ be such that vertex $g$ has color $2$ and all other vertices have color $1$.
    See \cref{fig:whard:weighted-mci-reduction-appendix} for a visualization of the construction.
    We remark that $G'$ can be constructed in polynomial time and has a vertex cover of size $3q+1$ by $\{g\} \cup \setdef{a_i,b^1_i,b^2_i}{i\in[q]}$.
    Hence the claimed lower bound holds assuming that the reduction preserves yes- and no-instances, which we prove in the following.

    \begin{figure}[ht]
        \centering
        \input{graphics/MCI_weighted_reduction}
        \caption{Local view on the construction for the reduction in \cref{thm:whard:weighted-vc}: two selection gadgets $W_i, W_{i'}$ with a conflict gadget for the $j$-th vertex in $V_i$ and the $j'$-th vertex in $V_{i'}$ between them.}
        \label{fig:whard:weighted-mci-reduction-appendix}
    \end{figure}
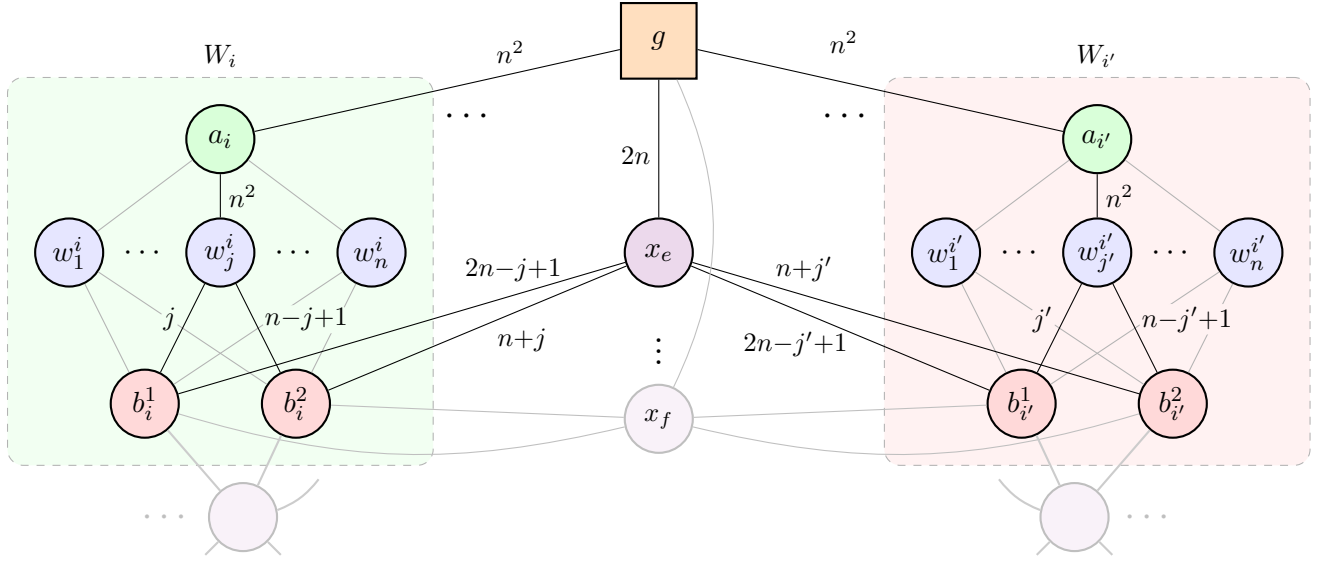

    \paragraph{Observation on Distances.}
    Before proving correctness, we observe the following properties of the distance between the constructed vertices in $G'$ for edges and vertices in $G$.
    In particular, for every edge $e=\{v^i_j,v^{i'}_{j'}\}$ of $G$ with $i\ne i'$, we have
    \begin{align}
        \dist_{G'}(x_e, w^i_j) &= \dist_{G'}(x_e, w^{i'}_{j'}) = 2n+1;\label{eq:whard:weighted-edge-incident} \\
        \dist_{G'}(x_e, w^i_h) &\le 2n \quad\text{for every } h\in[n]\setminus\{j\},\label{eq:whard:weighted-edge-not-incident-i}\\
        \dist_{G'}(x_e, w^{i'}_h) &\le 2n \quad\text{for every } h\in[n]\setminus\{j'\},\label{eq:whard:weighted-edge-not-incident-ip}\\
        \dist_{G'}(x_e, w^r_h) &> 3n \quad\text{for every } r\in[q]\setminus\{i,i'\}\text{ and }h\in[n].\label{eq:whard:weighted-edge-other-color}
    \end{align}

    For \eqref{eq:whard:weighted-edge-incident}, we argue for $w^i_j$; the case of $w^{i'}_{j'}$ is symmetric.
    The paths through $b^1_i$ and through $b^2_i$ both have length $2n+1$.
    These are shortest paths: any path using $g$ or $a_i$ uses an edge of weight $n^2>2n+1$; any path through another edge vertex must use at least three edges of weight at least $n+1$ before reaching a selector vertex; and any path that switches between $b^1_i$ and $b^2_i$ through some selector $w^i_h$ has length at least $(n+1)+(n+1)+1>2n+1$.

    For \eqref{eq:whard:weighted-edge-not-incident-i}, let $h\ne j$.
    If $h<j$, then the path through $b^1_i$ has length $(2n-j+1)+h\le 2n$.
    If $h>j$, then the path through $b^2_i$ has length $(n+j)+(n-h+1)\le 2n$.
    The proof of \eqref{eq:whard:weighted-edge-not-incident-ip} is identical.
    Finally, for \eqref{eq:whard:weighted-edge-other-color}, any path from $x_e$ to a selector in a group $W_r$ with $r\notin\{i,i'\}$ either goes through $g$, and hence has length more than $2n+n^2$, or goes through another edge vertex, which again requires at least three edges of weight at least $n+1$ before reaching $W_r$.

    \paragraph{Correctness: \MCIshort $\to$ \WMCSshort.}
    First suppose that $G$ contains an independent set $S\subseteq V(G)$ of size $q$. Let $S'\subseteq V(G')$ consist of $g$ and the vertex $w^i_j$ for every $v^i_j\in S$. We show that $S'$ is a solution to the {\WMCS} instance. Clearly, $|S'| = k = q+1$.
    As $g$ itself is contained in $S'$, it remains to argue that for every vertex $u\in V(G')\setminus S'$ there is a vertex $y\in S'\setminus \{g\}$ such that $\dist_{G'}(u,y) \le \dist_{G'}(u,g)$.
    Consider any $i\in [q]$ and let $v^i_s$ be the unique vertex in $S\cap V_i$ (as $G[V_i]$ is a clique, $S$ can only contain one vertex of each $V_i$). Note that $w^i_s\in S'$.
    \begin{itemize}
        \item $a_i$ has distance $n^2$ to both $g$ and $w^i_s$, so it is satisfied.
        \item $b^1_i$ and $b^2_i$ have distance $s \le n$ and $n-s+1 \le n$ to $w^i_s$, respectively. As $g$ has distance at least $2n$ to every other vertex, $b^1_i$ and $b^2_i$ are satisfied.
        \item The case for a vertex $w^i_j$ with $v^i_j\in V_i$ is similar. Its distance to $w^i_s$ (via $b^1_i$ or $b^2_i$) is at most $2n$ and thereby less than its distance to $g$.
    \end{itemize}

    The only remaining vertices are the ones representing an edge.
    Every such vertex $x_e$, for $e=\{v^i_j,v^{i'}_{j'}\}$, has distance $2n$ to $g$.
    As $S$ is an independent set, at least one endpoint of $e$ is not in $S$.
    Assume without loss of generality that $v^i_j\notin S$, and let $v^i_s$ be the unique vertex in $S \cap V_i$.
    Then $s\ne j$, $w^i_s\in S'$, and by \eqref{eq:whard:weighted-edge-not-incident-i} we have $\dist_{G'}(x_e, w^i_s) \le 2n$ as desired.

    \paragraph{Correctness: \WMCSshort $\to$ \MCIshort.}
    For the other direction, suppose there is a solution $S'$ to the {\WMCS} instance. We have $g\in S'$ as it is the only vertex of its color in $G'$.
    Thus, for all $i\in [q]$, vertex $a_i$ (having distance $n^2$ to $g$) is itself in $S'$ or some selector vertex $w^i_j$ is in $S'$.
    As $|S'\setminus \{g\}| \le k-1 = q$, the solution contains exactly one vertex from each set $\{a_i,w^i_1,\ldots,w^i_n\}$ and no other color-$1$ vertices.

    If $a_i \in S'$ for some $i\in [q]$, consider an alternative solution in which $a_i$ is replaced by an arbitrary selector $w^i_j$.
    This is still a solution: $a_i$ itself has distance $n^2$ to $w^i_j$, matching its distance to $g$; every other vertex in $W_i$ has distance at most $2n$ to $w^i_j$ and distance more than $2n$ to $g$; and no vertex outside $W_i$ can rely on $a_i$, since every path from such a vertex to $a_i$ is longer than the path from that vertex to $g$.

    Consequently, there is a solution $S'$ that consists of $g$ and precisely one selector vertex $w^i_{s_i}$ for each $i\in [q]$.
    Consider the corresponding $q$ vertices $v^1_{s_1}, \ldots, v^q_{s_q}$ in $G$ and suppose two of them, say $v^i_{s_i}$ and $v^{i'}_{s_{i'}}$, shared an edge $e$.
    Then $x_e$ has distance $2n$ to $g$ and is not selected, so some selected color-$1$ vertex must have distance at most $2n$ to $x_e$.
    By \eqref{eq:whard:weighted-edge-other-color}, such a selected vertex must lie in $W_i$ or $W_{i'}$.
    The only selected vertices in these two groups are $w^i_{s_i}$ and $w^{i'}_{s_{i'}}$, but both have distance $2n+1$ to $x_e$ by \eqref{eq:whard:weighted-edge-incident}.
    This contradiction shows that $v^1_{s_1}, \ldots, v^q_{s_q}$ form an independent set in $G$.
\end{proof}

%% file: graphics/MCI_weighted_reduction.tex
% adapted from code generated by Claude, Opus 4.6

\begin{tikzpicture}[
    % --- Node styles ---
    col1/.style={draw, circle, minimum size=9mm, inner sep=0pt, fill=blue!10, thick},
    col2/.style={draw, minimum size=10mm, inner sep=1pt, fill=orange!25, thick},
    anode/.style={col1, fill=green!15},
    bnode/.style={col1, fill=red!15},
    vnode/.style={col1, fill=blue!10},
    enode/.style={col1, fill=violet!15},
    every edge/.style={thick},
    weight/.style={midway, inner sep=1pt, rounded corners=5pt, font=\small},
    groupbox/.style={rounded corners=6pt, draw=gray!60, dashed, inner sep=10pt},
    >=stealth,
]
 
% =====================================================
% Global vertex g (color 2 -> square shape)
% =====================================================
\node[col2] (g) at (0, -2) {$g$};
 
% =====================================================
% Color group i (left side) — multiple v-vertices
% =====================================================
\begin{scope}[shift={(-5.8, -3.8)}]
    % a-nodes
    \node[anode] (ai) at (0, 0.5) {$a_i$};
 
    % Multiple v-vertices: first (1), highlighted (j), last (n)
    \node[vnode] (v1) at (-2, -1) {$w^i_1$};
    \node[vnode] (vj) at ( 0, -1) {$w^i_j$};
    \node[vnode] (vn) at ( 2, -1) {$w^i_n$};
    \node[font=\large] at (-1, -1) {$\cdots$};
    \node[font=\large] at ( 1, -1) {$\cdots$};

    % b-nodes
    \node[bnode] (b1i) at (-1, -3) {$b^1_i$};
    \node[bnode] (b2i) at ( 1, -3) {$b^2_i$};

    % --- Edges from a-nodes to g ---
    \draw (ai) -- node[weight, pos=0.7, above =3pt] {$n^2$} (g);

    % --- Edges from first/last selectors (dimmed) ---
    \draw[gray!60] (v1) -- (ai);
    \draw[gray!60] (v1) -- (b1i);
    \draw[gray!60] (v1) -- (b2i);
    \draw[gray!60] (vn) -- (ai);
    \draw[gray!60] (vn) -- (b1i);
    \draw[gray!60] (vn) -- (b2i);

    % --- Edges from highlighted selector w^i_j ---
    \draw (vj) -- node[weight, right = 2pt, pos=0.5] {$n^2$} (ai);
    \draw (vj) -- node[weight, fill=green!5, left=3pt, pos=0.4] {$j$} (b1i);
    \draw (vj) -- node[weight, fill=green!5, right=3pt, pos=0.4] {$n{-}j{+}1$} (b2i);

    % Group box
    \begin{scope}[on background layer]
        \node[groupbox, fill=green!5, fit=(ai)(b1i)(b2i)(v1)(vj)(vn), label={[font=\small]above: $W_i$}] {};
    \end{scope}
\end{scope}
 
% =====================================================
% Color group i' (right side) — multiple v-vertices
% =====================================================
\begin{scope}[shift={(5.8, -3.8)}]
    % a-nodes
    \node[anode] (aip) at (0, 0.5) {$a_{i'}$};
 
    % Multiple v-vertices: first (1), highlighted (j'), last (n)
    \node[vnode] (vl1) at (-2, -1) {$w^{i'}_1$};
    \node[vnode] (vl)  at ( 0, -1) {$w^{i'}_{j'}$};
    \node[vnode] (vln) at ( 2, -1) {$w^{i'}_n$};
    \node[font=\large] at (-1, -1) {$\cdots$};
    \node[font=\large] at ( 1, -1) {$\cdots$};

    % b-nodes
    \node[bnode] (b1ip) at (-1, -3) {$b^1_{i'}$};
    \node[bnode] (b2ip) at ( 1, -3) {$b^2_{i'}$};

    % --- Edges from a-nodes to g ---
    \draw (aip) -- node[weight, pos=0.6, above = 10pt] {$n^2$} (g);

    % --- Edges from first/last selectors (dimmed) ---
    \draw[gray!60] (vl1) -- (aip);
    \draw[gray!60] (vl1) -- (b1ip);
    \draw[gray!60] (vl1) -- (b2ip);
    \draw[gray!60] (vln) -- (aip);
    \draw[gray!60] (vln) -- (b1ip);
    \draw[gray!60] (vln) -- (b2ip);

    % --- Edges from highlighted selector w^{i'}_{j'} ---
    \draw (vl) -- node[weight, right=2pt, pos=0.5] {$n^2$} (aip);
    \draw (vl) -- node[weight, fill=red!5, left = 3pt, pos=0.4] {$j'$} (b1ip);
    \draw (vl) -- node[weight, fill=red!5, right= 3pt, pos=0.4] {$n{-}j'{+}1$} (b2ip);

    % Group box
    \begin{scope}[on background layer]
        \node[groupbox, fill=red!5, fit=(aip)(b1ip)(b2ip)(vl1)(vl)(vln), label={[font=\small]above: $W_{i'}$}] {};
    \end{scope}
\end{scope}
 
% =====================================================
% Edge vertex x_e (for edge e={v^i_j, v^{i'}_{j'}} in G)
% =====================================================
\node[enode] (vjl) at (0, -4.8) {$x_e$};
 
% Edge from v_{j,ell} to g
\draw (vjl) -- node[weight, left = 2pt, pos=0.45] {$2n$} (g);
 
% Edges from v_{j,ell} to b-nodes of color i (left)
\draw (vjl) -- node[weight, pos=0.25, above=4pt] {$2n{-}j{+}1$} (b1i);
\draw (vjl) -- node[weight, pos=0.35, below=6pt] {$n{+}j$} (b2i);
 
% Edges from v_{j,ell} to b-nodes of color i' (right)
\draw (vjl) -- node[weight, pos=0.25, above=4pt] {$n{+}j'$} (b2ip);
\draw (vjl) -- node[weight, pos=0.35, below=7pt] {$2n{-}j'{+}1$} (b1ip);
 
% =====================================================
% Second edge vertex x_f (dimmed, another edge
% between the same two color classes)
% =====================================================
\node[enode, fill=violet!5, draw=gray!50, font=\small] (vjplp) at (0, -7) {$x_f$};
 
% Dimmed edges from v_{j',ell'} to g
\draw[gray!50] (vjplp) to[bend right=25] (g);
 
% Dimmed edges from v_{j',ell'} to b-nodes of color i
\draw[gray!50] (vjplp) to[bend left=15] (b1i);
\draw[gray!50] (vjplp) -- (b2i);
 
% Dimmed edges from v_{j',ell'} to b-nodes of color i'
\draw[gray!50] (vjplp) -- (b1ip);
\draw[gray!50] (vjplp) to[bend right=15] (b2ip);
 
\node[font=\Large] at (0, -6) {$\vdots$};
 
% =====================================================
% Hints at edge vertices to other color classes
% =====================================================
% Hint edge vertex on left (edges from color i to other colors)
\node[enode, fill=violet!5, draw=gray!50, font=\small] (hintEL) at (-5.5, -8.3) {};
\draw[gray!40, thick] (hintEL) -- (b1i);
\draw[gray!40, thick] (hintEL) -- (b2i);
\draw[gray!40, thick] (hintEL) -- ++(.5, -.5);
\draw[gray!40, thick] (hintEL) -- ++(-.5, -.5);
\draw[gray!40, thick] (hintEL) to[bend right=15] ++(1, .5);

\node[font=\Large, gray!50] at (-6.5, -8.3) {$\cdots$};
 
% Hint edge vertex on right (edges from color i' to other colors)
\node[enode, fill=violet!5, draw=gray!50, font=\small] (hintER) at (5.5, -8.3) {};
\draw[gray!40, thick] (hintER) -- (b1ip);
\draw[gray!40, thick] (hintER) -- (b2ip);
\draw[gray!40, thick] (hintER) -- ++(.5, -.5);
\draw[gray!40, thick] (hintER) -- ++(-.5, -.5);
\draw[gray!40, thick] (hintER) to[bend left=15] ++(-1, .5);

\node[font=\Large, gray!50] at (6.5, -8.3) {$\cdots$};

% =====================================================
% Dots to indicate more color classes between the two
% =====================================================
\node[font=\Large] at (-2.5, -3) {$\cdots$};
\node[font=\Large] at ( 2.5, -3) {$\cdots$};
 
\end{tikzpicture}

%% file: Sections/tree-reduction.tex
Moving on, we argue that the exponential dependency on the number of colors in \cref{thm:twalg} is necessary.
To do so, we present a reduction from \VC, loosely inspired by a reduction of~\cite{fsttcs/BanikDMMNMRRS24}.

\begin{theorem}\label{thm:nphard:trees}
    {\MCS} cannot be solved in time $2^{o(n)}$ under the \textup{ETH},
    even on trees of treedepth at most $3$.
\end{theorem}

\begin{proof}
    Let $(G,q)$ be an instance of \VC, where
    $V(G) = \setdef{v_i}{i\in[n]}$ and $E(G) = \setdef{e_j}{j\in[m]}$.

    \paragraph{Construction.}
    First, introduce adjacent vertices $u,u'$ and give them a new common color $c^u$.
    For every $i\in[n]$, introduce a star on vertex set
    $X_i \coloneq \{x_i^1,x_i^2,\ell_i,\bar{\ell}_i\}$,
    centered at $x_i^1$. Furthermore, connect $x_i^1$ to $u$.
    Give $x_i^1,x_i^2$ a fresh common color $c_i^x$,
    and to $\ell_i,\bar{\ell}_i$ another fresh common color $c_i^\ell$.

    For every edge $e_j=v_a v_b$, add two length-two paths
    $(y_j^a,w_j^a,u)$ and $(y_j^b,w_j^b,u)$. Give $w_j^a,w_j^b$ a fresh
    common color $c_j^w$, and give $y_j^a$ and $y_j^b$ colors $c_a^x$ and
    $c_b^x$, respectively. All colors declared fresh are pairwise distinct.
    Let $Q_j \coloneq\{y_j^a,w_j^a,y_j^b,w_j^b\}$.

    This completes the construction of the graph.
    The resulting graph $T$ is a tree on $2+4n+4m=\bO(n+m)$ vertices.
    See \cref{fig:nphard:vertex-integrity-trees-appendix} for an illustration.
    Every vertex is at distance at most two from $u$, so $T$ has diameter at most four.
    Furthermore, $T - u$ is a forest whose every component is a star of order at most four,
    hence $T$ has vertex integrity at most five and treedepth at most three.
    Set $k \coloneq 1+2n+q+2m$.

    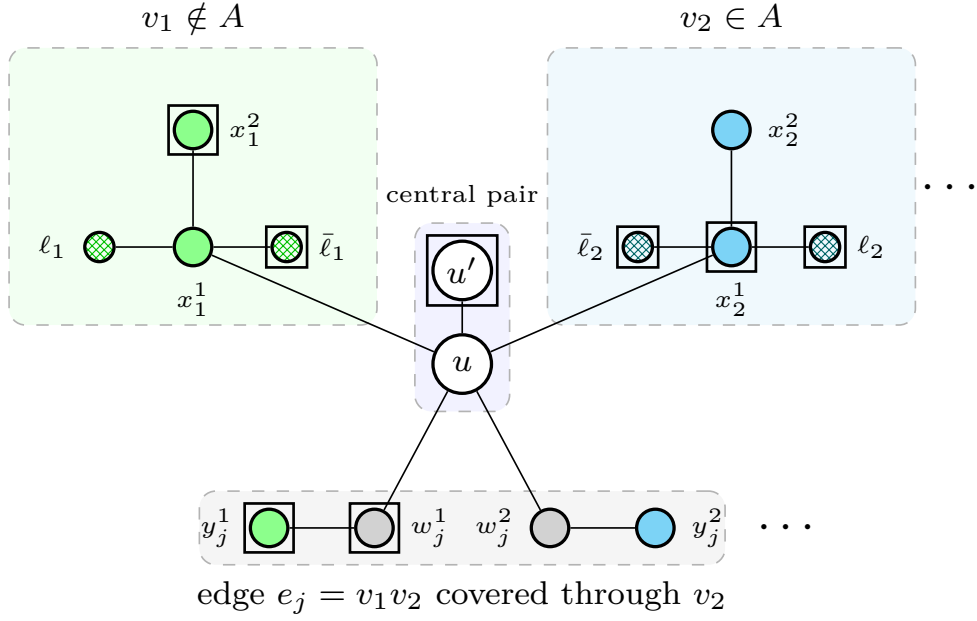
\begin{figure}
        \centering
        \resizebox{0.82\textwidth}{!}{\input{graphics/trees-reduction.tex}}
        \caption{Part of the construction in the proof of \cref{thm:nphard:trees}.
        Squared vertices are selected.
        %Fully labeled version of \cref{fig:nphard:vertex-integrity-trees}.
        }
        \label{fig:nphard:vertex-integrity-trees-appendix}
    \end{figure}

    \paragraph{Correctness: \VCshort $\to$ \MCSshort.}
    Let $A\subseteq V(G)$ be a vertex cover of $G$ of size at most $q$.
    We construct a set $S \subseteq V(T)$ as follows.
    Start with $u'\in S$. For each $i\in[n]$, add
    \[
        \begin{cases}
            x_i^1,\ell_i,\bar{\ell}_i, & v_i\in A,\\
            x_i^2,\bar{\ell}_i,        & v_i\notin A.
        \end{cases}
    \]
    For an edge $e_j=v_a v_b$, choose an endpoint in $A$, let $v_i\in\{v_a,v_b\}$ be its
    other endpoint, and add $w_j^i,y_j^i$ to $S$. Consequently,
    \[
        |S|=1+2n+|A|+2m\leq k.
    \]

    It remains to verify consistency. Every selected vertex is satisfied, so
    consider only the unselected vertices. The hub $u$ is satisfied by $u'$ at
    distance one. If $v_i\in A$, the only unselected vertex of $X_i$ is
    $x_i^2$, which is adjacent to the same-colored $x_i^1\in S$. If
    $v_i\notin A$, then $x_i^1$ is satisfied by the adjacent $x_i^2\in S$,
    while $\ell_i$ is at distance two from both $x_i^2$ and the same-colored
    $\bar{\ell}_i\in S$. No selection outside $X_i$ is closer to $\ell_i$.

    Finally, consider an edge $e_j=v_a v_b$, where $v_a \in A$, $w_j^b,y_j^b \in S$,
    and $w_j^a, y_j^a \notin S$. In that case, both $w_j^b$ and $x_a^1$ belong to $S$
    and are neighbors of $u$, thus satisfying $w_j^a$ and $y_j^a$ respectively.
    Consequently, every vertex is satisfied, and $S$ is consistent.

    \paragraph{Correctness: \MCSshort $\to$ \VCshort.}
    Let $S$ be a consistent subset of $T$ of size at most $k$,
    and define $A \coloneq \setdef{v_i \in V(G)}{x_i^1 \in S}$.

    The central color forces at least one of $u,u'$ into $S$.
    For every $i\in[n]$, at least one
    of $\ell_i,\bar{\ell}_i$ is selected, as their color occurs nowhere else.
    At least one of $x_i^1,x_i^2$ is selected as well: otherwise $x_i^1$ would
    see either $\ell_i$ or $\bar{\ell}_i$ at distance one, whereas every selected vertex
    of color $c_i^x$ would be at distance at least three. Finally, if
    $v_i\in A$, then both $\ell_i,\bar{\ell}_i$ are selected, as otherwise
    the unselected leaf would see $x_i^1$ at distance one.
    Thus every vertex gadget contributes at least two vertices, and each
    gadget indexed by a vertex of $A$ contributes at least three. Their total
    contribution is therefore at least $2n+|A|$.

    Consider an edge $e_j=v_a v_b$. Since $c_j^w$ occurs only on
    $w_j^a,w_j^b$, some $w_j^i$, $i \in \{a,b\}$, is selected. Its adjacent vertex
    $y_j^i$ must then also be selected, as it is a leaf attached to $w_j^i$ with different color.
    Thus every edge gadget contributes at least two vertices.

    Suppose in addition that neither endpoint of $e_j$ belongs to $A$. For
    each $r\in\{a,b\}$, every selected vertex of color $c_r^x$ outside $Q_j$
    is at distance at least four from $y_j^r$: the only variable-gadget vertex
    at distance three is the unselected $x_r^1$, while $x_r^2$ and all other
    occurrence vertices are at distance four. On the other hand, the selected
    vertex $w_j^i$ is at distance at most three from both occurrence
    vertices. Hence both $y_j^a,y_j^b$ must be selected, and $Q_j$ contributes
    at least three vertices.

    Let $\kappa$ be the number of edges with neither endpoint in $A$. The central
    pair contributes at least one vertex, so
    \[
        1+(2n+|A|)+(2m+\kappa)\leq |S|\leq k.
    \]
    Hence $|A|+\kappa \leq q$. Add to $A$ one arbitrary endpoint of every edge not
    already covered by $A$. The resulting set is a vertex cover of $G$ of size
    at most $|A|+\kappa \leq q$.

    \paragraph{Wrap-up.}
    Correctness follows from the preceding paragraphs.
    Furthermore, a $2^{o(|V(T)|)}$-time algorithm for \MCSshort would yield a $2^{o(n+m)}$-time algorithm
    for \VC, contradicting the ETH.
\end{proof}

We remark that the construction underlying the proof of \cref{thm:nphard:trees}
in fact produces a slightly stronger result: instead of treedepth, it also bounds the deletion
distance to connected components of bounded size.
More precisely, the result would hold even when bounding a graph-theoretic measure
called \emph{vertex integrity}~\cite{LampisM24,GimaHKMOO25}.
Moreover, the treedepth bound in \cref{thm:nphard:trees} is tight: every graph of treedepth at most two is a star forest,
and {\MCS} can be solved in polynomial time on such graphs as it suffices to consider, independently for each
star, whether its center is selected.

%% file: graphics/trees-reduction.tex
\usetikzlibrary{patterns}
\begin{tikzpicture}[
    pnode/.style={draw, circle, inner sep=0pt, minimum size=3.2mm, thick},
    central/.style={pnode, fill=white, minimum size=5mm},
    varA/.style={pnode, fill=green!45},
    varB/.style={pnode, fill=cyan!45},
    edgecol/.style={pnode, fill=gray!35},
    stab/.style={pnode, fill=white, minimum size=2.5mm,
        pattern=crosshatch},
    stabA/.style={stab, pattern color=green!75!black},
    stabB/.style={stab, pattern color=teal!80!black},
    groupbox/.style={rounded corners=4pt, draw=gray!60, dashed,
        inner sep=4pt},
    selected/.style={draw=black, rectangle, line width=0.65pt,
        inner sep=0.35mm},
    every edge/.style={thick},
    font=\scriptsize,
]

% Central edge and hub.
\node[central] (u) at (0,0) {$u$};
\node[central] (up) at (0,0.8) {$u'$};
\draw (u) -- (up);
\begin{scope}[on background layer]
    \node[groupbox, fill=blue!5, fit=(u)(up),
        label={[font=\tiny]above:central pair}] {};
\end{scope}

% Gadget for v_1 outside the vertex cover.
\node[varA, label={[font=\tiny]below:$x_1^1$}] (A1) at (-2.3,1.0) {};
\node[varA, label={[font=\tiny]right:$x_1^2$}] (A2) at (-2.3,2.0) {};
\draw (A1) -- (A2);
\draw (A1) -- (u);
\node[stabA, label={[font=\tiny]left:$\ell_1$}] (Al) at (-3.1,1.0) {};
\node[stabA, label={[font=\tiny]right:$\bar\ell_1$}] (Abl) at (-1.5,1.0) {};
\draw (Al) -- (A1) -- (Abl);
\begin{scope}[on background layer]
    \node[groupbox, fill=green!5, inner xsep=18pt, inner ysep=15pt,
        fit=(Al)(A2)(Abl),
        label={[font=\scriptsize]above:$v_1\notin A$}] {};
\end{scope}

% Gadget for v_2 in the vertex cover.
\node[varB, label={[font=\tiny]below:$x_2^1$}] (B1) at (2.3,1.0) {};
\node[varB, label={[font=\tiny]right:$x_2^2$}] (B2) at (2.3,2.0) {};
\draw (B1) -- (B2);
\draw (B1) -- (u);
\node[stabB, label={[font=\tiny]right:$\ell_2$}] (Bl) at (3.1,1.0) {};
\node[stabB, label={[font=\tiny]left:$\bar\ell_2$}] (Bbl) at (1.5,1.0) {};
\draw (Bl) -- (B1) -- (Bbl);
\begin{scope}[on background layer]
    \node[groupbox, fill=cyan!5, inner xsep=18pt, inner ysep=15pt,
        fit=(Bbl)(B2)(Bl),
        label={[font=\scriptsize]above:$v_2\in A$}] {};
\end{scope}

\node[font=\large] at (4.2,1.5) {$\cdots$};

% Edge gadget e_1=v_1v_2, covered through v_2.
\node[varA,label={[font=\tiny]west:$y_j^1$}] (ya) at (-1.65,-1.4) {};
\node[edgecol,label={[font=\tiny]east:$w_j^1$}] (wa) at (-0.75,-1.4) {};
\node[edgecol,label={[font=\tiny]west:$w_j^2$}] (wb) at (0.75,-1.4) {};
\node[varB,label={[font=\tiny]east:$y_j^2$}] (yb) at (1.65,-1.4) {};
\draw (ya) -- (wa) -- (u) -- (wb) -- (yb);
\begin{scope}[on background layer]
    \node[groupbox, fill=gray!8, inner xsep=12pt, fit=(ya)(wa)(wb)(yb),
        label={[font=\scriptsize]below:edge $e_j=v_1 v_2$ covered through $v_2$}] {};
\end{scope}
\node[font=\large] at (2.8,-1.4) {$\cdots$};

% Selected vertices.
\node[selected, fit=(up)] {};
\node[selected, fit=(A2)] {};
\node[selected, fit=(Abl)] {};
\node[selected, fit=(B1)] {};
\node[selected, fit=(Bl)] {};
\node[selected, fit=(Bbl)] {};
\node[selected, fit=(wa)] {};
\node[selected, fit=(ya)] {};

\end{tikzpicture}

%% file: Sections/vc-algos.tex
\section{Parameterization by Vertex Cover Number}
In this section we consider the parameterization by vertex cover number.
As the main result of this section, we significantly improve upon the previous $\vc^{\bO(\vc)} n^{\bO(1)}$ algorithm for \MCSshort~\cite{BanikPRRS26}, and present one with running in time $5^{\vc} n^{\bO(1)}$.
%Moreover, from known lower bounds, it follows that the single-exponential dependence on the first term of our algorithm is unavoidable under the ETH.
We then extend our ideas to solve
%using similar ideas, we extend our algorithm to the weighted setting, where we show that
\WMCSshort in time $n^{\bO(\vc)}$, matching the lower bound of \cref{thm:whard:weighted-vc}.

We first  fix some notation.
For non-empty sets $M,S\subseteq V(G)$, the \emph{distance vector from $M$ to $S$} is
$    \bigl(\dist(v,S)\bigr)_{v\in M}.
$

On a high level, our algorithms build on the approach of \cite[Theorem~1.9]{BanikPRRS26}:
enumerate candidate distance vectors from a vertex cover $M$ to a solution, and optimize over the solutions realizing each candidate.
For \MCSshort, we achieve a single-exponential running time by combining a subroutine for fixed distance vectors with a bound on the set of candidate vectors.
%In order to achieve the improved running time we introduce several new ideas and make use of a more refined analysis. For unweighted graphs, we combine a subroutine for a fixed distance vector with a small set of candidate vectors.
For \WMCSshort---where the number of relevant vectors may depend on $n$---we instead give a more direct algorithm based on the same distance-vector viewpoint and a direct enumeration of small witnesses. We capture the unifying insight into distance vectors in the following lemma.

\begin{lemma}\label{lem:vc:weighted-fixed-distance-vector}
    Let $(G,w,k,\col)$ be an instance of \WMCS.
    Furthermore, let $M \subseteq V(G)$ be a vertex cover of $G$,
    and let $\mathbf{d}=(d_u)_{u\in M}$,
    where $g$ entries of $\mathbf{d}$ are zero.
    Then a minimum-size consistent subset of $G$ whose distance vector from
    $M$ equals $\mathbf{d}$ can be found, or its non-existence established,
    in $\bO(2^{|M|-g}\cdot |M|\cdot n^2)$ time.
\end{lemma}

\begin{proof}
    We say that a set $S \subseteq V(G)$ \emph{realizes} $\mathbf{d}$ if
    it satisfies $\dist(u,S)=d_u$ for every $u\in M$.

    Let $I \coloneq V(G) \setminus M$.
    Since $M$ is a vertex cover, $I$ is an independent set and
    $N(v) \subseteq M$ for every $v \in I$.
    We compute $\dist(u,v)$ for all $u \in M$ and $v \in V(G)$ by running
    Dijkstra's algorithm once from each vertex of $M$.

    Define $M_0 \coloneq \setdef{u\in M}{d_u=0}$
    and $M_{>0} \coloneq M\setminus M_0$.
    Moreover, define
    \[
        I_\text{out} \coloneq
        \setdef{v\in I}{
            \exists \, u \in M \colon \dist(u,v) < d_u
        }.
    \]
    We reject if $\dist(u,M_0)<d_u$ for some $u\in M$.
    Otherwise, let
    \[
        A \coloneq M_0\cup(I\setminus I_\text{out}).
    \]

    \begin{claim}\label{claim:vc:admissible-range}
        If $S$ realizes $\mathbf{d}$, then $M_0\subseteq S\subseteq A$.
        Conversely, every $S \subseteq A$ satisfies $\dist(u,S)\geq d_u$
        for every $u\in M$.
    \end{claim}

    \begin{claimproof}
        Since all edge weights are positive, $\dist(u,S)=0$ holds if and only
        if $u\in S$. Thus a set realizing $\mathbf{d}$ satisfies
        $S\cap M=M_0$.
        Also, no vertex of $I_\text{out}$ can be selected: if
        $v\in I_\text{out}\cap S$, then some $u\in M$ has
        $\dist(u,S)\leq \dist(u,v)<d_u$.
        Hence $M_0\subseteq S\subseteq A$.

        Conversely, let $S \subseteq A$.
        Since $\dist(u,S_1) \geq \dist(u,S_2)$ whenever $S_1\subseteq S_2$, it suffices to prove
        the claim for $S=A$.
        If $u \in M_0$, then $\dist(u,A) = 0 = d_u$.
        Alternatively, let $u \in M_{>0}$.
        The initial rejection test guarantees that $\dist(u,M_0) \geq d_u$,
        while by the definition of $I_\text{out}$,
        $\dist(u,v) \geq d_u$ for every $v \in I \setminus I_\text{out}$.
        This completes the proof.
    \end{claimproof}

    For each $v\in I$, define
    \[
        \delta_v \coloneq
        \min_{u\in N(v)}
        \bigl\{w(\{v,u\})+d_u\bigr\},
    \]
    and let
    \[
        N^*(v) \coloneq
        \setdef{u\in N(v)}{
            w(\{v,u\})+d_u=\delta_v
        }.
    \]
    Intuitively, $\delta_v$ denotes the distance of $v$ to $S$ when $v \notin S$ (so every path to $S$ goes via $M$),
    while $N^*(v)$ denotes the set of its neighbors through which $v$ realizes $\delta_v$.
    For each $v \in V(G)$, define
    \[
        W(v) \coloneq
        \setdef{u\in M}{\dist(u,v)=d_u}.
    \]
    Thus $W(v)$ is the set of cover vertices that see $v$ at the same distance that they see the solution $S$.
    Thus, if $v\in S$, then every vertex in $W(v)$ sees color $\col(v)$ in $S$. We hence say that $W(v)$ is \emph{witnessed} by selecting $v$.
    For $V' \subseteq V(G)$, write
    \[
        W(V') \coloneq \bigcup_{v \in V'}W(v).
    \]

    \begin{claim}\label{claim:vc:improved-color-characterization}
        Let $M_0\subseteq S\subseteq A$, and define
        $S_x \coloneq S\cap V_x$ for every color $x\in[c]$.
        Then $S$ is consistent and realizes $\mathbf{d}$ if and only if, for
        every color $x\in[c]$:
        \begin{enumerate}[label=(\roman*),itemsep=2pt]
            \item all cover vertices of color $x$ are satisfied via a vertex of color $x$, that is,
            \[
                V_x\cap M \subseteq W(S_x);
            \]\label{item:vc:improved-own-color}
            \item all independent vertices of color $x$ that are not selected see $x$ via a neighbor in the cover,
            that is, for every $v\in V_x\cap(I\setminus S_x)$,
            \[
                N^*(v)\cap W(S_x)\neq\varnothing.
            \]\label{item:vc:improved-independent}
        \end{enumerate}
    \end{claim}

    \begin{claimproof}
        Suppose first that Conditions~\ref{item:vc:improved-own-color} and~\ref{item:vc:improved-independent} hold.
        By \cref{claim:vc:admissible-range}, $\dist(u,S)\geq d_u$ for every
        $u\in M$.
        For $u\in M$, Condition~\ref{item:vc:improved-own-color} with
        $x=\col(u)$ gives a vertex $s\in S_x$ such that
        $u\in W(s)$, that is, $\dist(u,s)=d_u$.
        Hence $\dist(u,S)=d_u$, and $u$ is satisfied.

        Every selected vertex of $I$ is satisfied by itself.
        Consider an unselected vertex $v\in I$ and let $x \coloneq \col(v)$.
        By Condition~\ref{item:vc:improved-independent}, there are vertices
        $u^\star \in N^*(v)$ and $s\in S_x$ such that $u^\star \in W(s)$.
        Every path from $v$ to $S$ starts with a neighbor $u \in N(v) \subseteq M$,
        and we already proved $\dist(u,S)=d_u$. Hence
        \[
            \dist(v,S)
            \geq
            \min_{u\in N(v)}\{w(\{v,u\})+d_u\}
            =
            \delta_v.
        \]
        On the other hand, the path from $v$ to $s$ through $u^\star$ gives
        \begin{align*}
            \dist(v,S)
            &\leq \dist(v,s)\\
            &\leq w(\{v,u^\star\})+\dist(u^\star,s)\\
            &=w(\{v,u^\star\})+d_{u^\star}
            =\delta_v.
        \end{align*}
        It follows that $\dist(v,S) = \delta_v = \dist(v,s)$,
        and since $s$ has color $x$, the vertex $v$ is satisfied.

        Conversely, suppose that $S$ is consistent and realizes $\mathbf{d}$.
        Condition~\ref{item:vc:improved-own-color} follows because every
        vertex of $M$ is satisfied at distance $d_u$ by a selected vertex of
        its own color.

        Now consider $v\in V_x\cap(I\setminus S_x)$.
        Since $v$ is satisfied, there is an $s\in S_x$ such that
        $\dist(v,s)=\dist(v,S)$.
        As $S$ realizes $\mathbf{d}$, every neighbor $u \in N(v)$ satisfies
        $\dist(u,S)=d_u$.
        Since $I$ is independent, $N(v) \subseteq M$, and
        \begin{align*}
            \dist(v,s)
            &= \dist(v,S)\\
            &= \min_{u \in N(v)}\{w(\{v,u\})+\dist(u,S)\}\\
            &= \min_{u\in N(v)}\{w(\{v,u\})+d_u\}\\
            &= \delta_v.
        \end{align*}
        Let $u \in N^*(v)$, that is, $u$ is a neighbor of $v$ attaining the minimum in the last expression.
        Then,
        \begin{align*}
            \dist(v,s)
            &= w(\{v,u\}) + \dist(u,s)\\
            &\ge w(\{v,u\})+\dist(u,S)\\
            &= w(\{v,u\})+d_u\\
            &= \delta_v.
        \end{align*}
        Since $\dist(v,s)=\delta_v$, we have equality throughout, and in particular
        $\dist(u,s)=d_u$, which means $u\in W(s)$.
        Consequently, $u \in N^*(v) \cap W(s)$,
        and Condition~\ref{item:vc:improved-independent} holds.
    \end{claimproof}

    The two conditions of
    \cref{claim:vc:improved-color-characterization} are local to the set
    $S_x$ of one color. Since the color classes partition $V(G)$, the
    objective is additive over the colors. We can therefore optimize each
    non-empty color class independently and take the union of the selected
    sets.

    \paragraph{Per-color subroutine.}
    Fix a non-empty color class $V_x$.
    The vertices of $V_x\cap M_0$ belong to every subset realizing $\mathbf{d}$.
    Let
    \[
        Z_x \coloneq W(V_x \cap M_0)
    \]
    be the set of cover vertices already witnessed by them.
    Note that a vertex of $M_0$ can be witnessed at distance zero only by itself,
    thus
    \[
        Z_x \cap M_0 = V_x \cap M_0.
    \]

    For every subset of $M_{>0}$, we compute a minimum-cardinality subset of
    $V_x \cap (I \setminus I_{\textrm{out}})$ that witnesses it.
    In particular, we build a \textsc{Set Cover} table as follows.
    For every $U \subseteq M_{>0}$, let $\cover_x(U)$ be the
    minimum size of a set
    \[
        Y_x \subseteq V_x \cap(I \setminus I_\text{out})
    \]
    such that $U \subseteq W(Y_x)$.
    If no such set exists, let $\cover_x(U) \coloneq \infty$.
    The values are computed by
    \[
        \cover_x(\varnothing)=0
    \]
    and, for non-empty $U\subseteq M_{>0}$,
    \[
        \cover_x(U)
        =
        1+
        \min_{\substack{
            v \in V_x \cap (I\setminus I_\text{out})\\
            W(v) \cap U \neq \varnothing
        }}
        \cover_x(U \setminus W(v)).
    \]
    An empty minimum is interpreted as $\infty$.
    Since $W(v) \cap U \neq \varnothing$, the recursive call is on a strict
    subset of $U$, so the table can be filled in increasing order of $|U|$.
    We store, for each finite entry, a set attaining the minimum.

    We now enumerate \emph{certificates} $C \subseteq M_{>0}$. The intended meaning is that for a solution $S$ we have $W(S_x\cap I)=C$, that is,
    the selected independent vertices of color $x$ witness precisely set $C$ in $M$. Thus every vertex $u \in Z_x \cup C$ satisfies $\dist(u,S)=d_u = \dist(u, S_x)$.
    % \simon{I adapted the definition and dropped the "while they may also witness additional vertices of $M_{>0}$" here because indeed we are looking for the $C$ where it witnessed exactly, right? It was this half-sentence which cause most of my confusion, since we speak here about intended meaning, and the intention is that $C$ covers it perfectly. Then also the sentence below the definition of $Q$ is correct. Please check!}
    % must \emph{witness} every vertex in
    % $Z_x \cup C$ \textcolor{blue}{(that is, every vertex $u \in Z_x \cup C$ satisfies $\dist(u,S)=d_u = \dist(u, S_x)$)}, while they may also witness additional vertices of $M_{>0}$.

    We only consider certificates satisfying
    \[
        V_x \cap M_{>0}\subseteq Z_x\cup C,
    \]
    because every cover vertex of color $x$ must be satisfied.

    For such a certificate, define
    \[
        Q_x^C \coloneq
        \setdef{v\in V_x\cap I}{
            N^*(v)\cap(Z_x\cup C)=\varnothing
        }.
    \]
    These are exactly the independent vertices of color $x$ that would fail
    Condition~\ref{item:vc:improved-independent} unless they were selected.
    Thus set
    \[
        F_x^C \coloneq (V_x\cap M_0)\cup Q_x^C
    \]
    to be the set of the vertices forced to be selected in this case.
    If $Q_x^C\cap I_\text{out}\neq\varnothing$, discard $C$, since it forces
    a forbidden vertex.

    The forced vertices may already witness part of $C$. Let the part of $C$ that is not yet witnessed by forced vertices be
    \[
        U_x^C \coloneq C\setminus W(F_x^C).
    \]
    To cover $U_x^C$, we use our precomputed table.
    If $\cover_x(U_x^C)$ is finite, let $Y_x^C$ be a stored set attaining $\cover_x(U_x^C)$ and set
    \[
         S_x^C \coloneq F_x^C \cup Y_x^C
    \]
   as well as $\cost_x(C) \coloneq |S_x^C| = |F_x^C| + |Y_x^C|$;
   the latter equality being due to the fact that $W(F_x^C) \cap C = \varnothing$, thus $F_x^C \cap Y_x^C = \varnothing$.
   We remark that $Y_x^C$ might cover vertices outside $Z_x\cup C$ as well, but in this case there is a certificate $C'\supset C$ of the same cost such that exactly the vertices in $Z_x\cup C'$ are witnessed.

    We retain a certificate of minimum cost and denote its set by
    $S_x$. If no certificate has finite cost for some color, we reject
    $\mathbf{d}$.

    \begin{claim}\label{claim:vc:certificate-soundness}
        For every certificate $C$ that is not discarded and has finite cost,
        the set $S_x^C$ satisfies
        Conditions~\ref{item:vc:improved-own-color}
        and~\ref{item:vc:improved-independent}.
    \end{claim}

    \begin{claimproof}
        We have $S_x^C\subseteq V_x\cap A$ and
        $V_x\cap M_0\subseteq S_x^C$.
        Moreover, every vertex of $Z_x\cup C$ is witnessed by $S_x^C$:
        the forced cover vertices witness $Z_x$, the forced independent vertices witness
        $C\setminus U_x^C$, and $Y_x^C$ witnesses $U_x^C$.
        Since
        \[
            V_x\cap M
            =
            (V_x\cap M_0)\cup(V_x\cap M_{>0})
            \subseteq Z_x\cup C,
        \]
        Condition~\ref{item:vc:improved-own-color} follows.

        Now let $v\in V_x\cap(I\setminus S_x^C)$.
        Then $v\notin Q_x^C$, so
        \[
            N^*(v)\cap(Z_x\cup C)\neq\varnothing.
        \]
        Since $Z_x\cup C\subseteq W(S_x^C)$, this implies
        Condition~\ref{item:vc:improved-independent}.
    \end{claimproof}

    \begin{claim}\label{claim:vc:certificate-optimality}
        The retained set $S_x$ has minimum size among all sets
        $S'_x\subseteq V_x\cap A$ such that $S'_x \cap M = V_x \cap M_0$ and
        Conditions~\ref{item:vc:improved-own-color}
        and~\ref{item:vc:improved-independent} hold with $S'_x$ in place of
        $S_x$.
    \end{claim}

    \begin{claimproof}
        Let $S'_x$ be any such feasible set, and put
        \[
            C' \coloneq W(S'_x \cap I).
        \]
        Since $S_x' \cap M = V_x \cap M_0$, we have $Z_x = W(S'_x \cap M)$, thus $W(S_x') = Z_x \cup C'$.

        By Condition~\ref{item:vc:improved-own-color},
        $V_x\cap M_{>0}\subseteq W(S_x') = Z_x\cup C'$, so $C'$ is a valid
        certificate.

        We next show that $Q_x^{C'}\subseteq S'_x$.
        If some $v\in Q_x^{C'}$ were not selected, then
        Condition~\ref{item:vc:improved-independent} would imply
        \[
            N^*(v)\cap W(S'_x)
            =
            N^*(v)\cap(Z_x\cup C')
            \neq\varnothing,
        \]
        contradicting the definition of $Q_x^{C'}$.
        Hence $F_x^{C'}\subseteq S'_x$, and in particular
        $Q_x^{C'}\cap I_\text{out}=\varnothing$; the certificate is not
        discarded.

        Every vertex of $U_x^{C'}$ is witnessed by a vertex of
        $S'_x\setminus F_x^{C'}$. Moreover,
        \[
            S'_x\setminus F_x^{C'}
            \subseteq V_x\cap(I\setminus I_\text{out}),
        \]
        so this set is admissible for the table defining $\cover_x$.
        Therefore
        \[
            \cover_x(U_x^{C'})
            \leq |S'_x\setminus F_x^{C'}|,
        \]
        and, as $F_x^{C'}\subseteq S_x'$,
        \[
            \cost_x(C')
            \leq
            |F_x^{C'}|+|S'_x\setminus F_x^{C'}|
            =
            |S'_x|.
        \]
        The retained certificate has minimum cost, and
        \cref{claim:vc:certificate-soundness} proves that its set is feasible.
        Since retained sets have size equal to their cost, $S_x$ is minimum.
    \end{claimproof}

    After processing every color, let
    \[
        S \coloneq \bigcup_{x\in[c]}S_x.
    \]
    By
    \cref{claim:vc:improved-color-characterization,claim:vc:certificate-soundness},
    the set $S$ is consistent and realizes $\mathbf{d}$.
    By \cref{claim:vc:certificate-optimality} and additivity over color
    classes, it has minimum size among all consistent sets realizing
    $\mathbf{d}$.

    \paragraph{Running time.}
    The distance preprocessing takes $\bO(|M|\cdot(m+n\log n))$ time.
    For each of the $c$ colors, we perform the following.
    First, $Z_x \coloneq W(V_x \cap M_0)$ can be computed in $\bO(|M|^2)$ time via computing $V_x \cap M_0$ in $\bO(|M|)$ time and computing $W(\cdot)$ via $\bO(|M_0| \cdot |M|)$ tests.
    We then fill the \textsc{Set Cover} table, which has $2^{|M_{>0}|}$ entries
    and tests at most $n$ candidates $v$ per entry, and for each of them computes $U\setminus W(v)$ in $\bO(|M|)$ time.
    Thus this step runs in time $\bO(2^{|M_{>0}|} \cdot |M| \cdot n)$.
    We enumerate $2^{|M_{>0}|}$ certificates. For each of them, the test whether $V_x\cap M_{>0} \subseteq Z_x \cup C$, the computation of $Q_x^C$ and $F_x^C$,
    as well as the test against $I_{\textup{out}}$ and the computation of $W(F_x^C)$ and $U_x^C$ can all be done in $\bO(|M| \cdot n)$ time.
    Thus, the subroutine takes $\bO(2^{|M_{>0}|} \cdot |M| \cdot n)$ time per non-empty color class.
    Since $|M_{>0}|=|M|-g$, the total running time is therefore $\bO(2^{|M|-g} \cdot |M| \cdot c \cdot n)$ plus the preprocessing time. For simplicity, we estimate the total by $\bO(2^{|M|-g} \cdot |M| \cdot n^2)$.
\end{proof}

\paragraph{Unweighted graphs.}
In an unweighted graph, all relevant distance vectors can be recovered from only $3^\vc$ signatures, representing whether each vertex cover vertex has distance $0$, $1$, or more than $1$ to the solution.
Combining this enumeration with \cref{lem:vc:weighted-fixed-distance-vector} yields the following FPT algorithm.

\begin{theorem}\label{thm:vc:unweighted-fpt}
    {\MCS} can be solved in time $\bO(5^\vc\cdot\vc\cdot n^2)$.
\end{theorem}

\begin{proof}
    Consider an instance $(G,k,\col)$ of \MCS.
    We first compute a minimum vertex cover $M$ of $G$ in time $\bO(2^\vc \cdot n)$ using the standard branching algorithm,
    where $|M|=\vc$.

    The key observation is that it suffices to record, for each vertex of $M$, whether its distance to a hypothetical solution is $0$, $1$, or at least $2$.
    We represent the last possibility by the symbol $\star$ and define the \emph{signature space} $\Sigma$
    to be the set of all functions from $M$ to $\{0,1,\star\}$.

    We first run breadth-first search from every vertex of $M$.
    This computes all distances with one endpoint in $M$ in
    $\bO(\vc\cdot(n+m))$ time.

    Fix a signature $\sigma \colon M \to \{0,1,\star\}$, and let
    \[
        M_a \coloneq \setdef{u \in M}{\sigma(u)=a}
        \qquad\text{for }a\in\{0,1,\star\}.
    \]
    We construct a numerical vector $\mathbf{d}^{\sigma} = (d^{\sigma}_u)_{u \in M}$ by setting
    $d^{\sigma}_u \coloneq 0$ for $u\in M_0$,
    $d^{\sigma}_u \coloneq 1$ for $u\in M_1$, and, for every $u\in M_\star$,
    \[
        d^{\sigma}_u \coloneq
        \min\left\{
            \min_{u' \in M_0}\dist(u,u'),\
            \min_{u' \in M_1}\bigl(\dist(u,u')+1\bigr)
        \right\}.
    \]
    A minimum over an empty set is understood to be $\infty$.
    If an entry indexed by $M_\star$ is less than $2$ or equal to $\infty$, we discard $\sigma$; otherwise, we retain $\mathbf{d}^{\sigma}$.
    In this way, all candidate vectors can be generated in
    $\bO(3^\vc\cdot\vc^2)$ time.

    We show that, for every consistent subset, this enumeration retains its distance vector.
    Let $S$ be a consistent subset, and define $\sigma_S\in\Sigma$ according to whether $\dist(u,S)$ is $0$, $1$, or at least $2$ for each $u\in M$.
    For $a\in\{0,1,\star\}$, let
    \[
        M_a^S \coloneq \setdef{u\in M}{\sigma_S(u)=a}.
    \]
    Fix $v \in M_\star^S$ and a shortest path from $v$ to $S$.
    Let $r \notin S$ be its penultimate vertex and $s\in S$ its final vertex.
    Since $M$ is a vertex cover, at least one endpoint of the edge $\{r,s\}$ belongs to $M$.
    If $s\in M$, then $s\in M_0^S$, so the first minimum is at most $\dist(v,S)$.
    Otherwise, $r\in M$ and $\dist(r,S)=1$, so $r\in M_1^S$ and the second minimum is at most $\dist(v,S)$.
    Hence
    \[
        d^{\sigma_S}_v\leq\dist(v,S).
    \]
    Conversely, every vertex of $M_0^S$ belongs to $S$, and every vertex of $M_1^S$ has a neighbor in $S$.
    Every term in the two minima therefore gives the length of a path from $v$ to $S$, implying
    \[
        d^{\sigma_S}_v\geq\dist(v,S).
    \]
    Thus $d^{\sigma_S}_v=\dist(v,S)$ for every $v\in M_\star^S$.
    Equality holds by definition on $M_0^S\cup M_1^S$, so
    $\mathbf{d}^{\sigma_S}$ is precisely the distance vector from $M$ to $S$ and is not discarded.

    We apply \cref{lem:vc:weighted-fixed-distance-vector}, with all edge weights equal to $1$, to every retained candidate vector.
    The distances computed above are reused in every application.
    If a signature assigns $0$ to exactly $g$ vertices, then its vector has exactly $g$ zero entries and is processed in
    $\bO(2^{\vc-g}\cdot\vc\cdot n^2)$ time.
    There are $\binom{\vc}{g}2^{\vc-g}$ such signatures: we choose the $g$ zero entries and assign either $1$ or $\star$ to every remaining entry.
    By the binomial theorem,
    \[
        \sum_{g=0}^{\vc}
            \binom{\vc}{g}2^{\vc-g}2^{\vc-g}
        =
         \sum_{g=0}^{\vc}
            \binom{\vc}{g}4^{\vc-g}1^{g}
        =
        5^\vc.
    \]
    The total running time is therefore
    $\bO(5^\vc\cdot\vc\cdot n^2)$.
    The preprocessing and enumeration times are dominated by this bound.
    Since the candidates include the distance vector of an optimal consistent subset, returning the smallest solution found proves the theorem.
\end{proof}

We remark that the exponential dependency in $\vc$ is necessary under the ETH.
Indeed, the reduction of~\cite[Section 2]{fsttcs/BanikDMMNMRRS24} starts from an instance of \textsc{Dominating Set}
and produces an instance of \MCSshort by introducing an additional universal vertex and using only two colors.
Thus, the known ETH lower bound for \textsc{Dominating Set} transfers to \MCSshort and yields an $2^{o(n)}$ lower bound under the ETH,
in particular excluding any algorithm of running time $2^{o(\vc)} n^{\bO(1)}$.

\paragraph{Weighted graphs.}
For \WMCSshort, the number of relevant distance vectors may depend on $n$ (consider, e.g., a star with pairwise distinct edge weights). Nevertheless, we show that the number of candidates can still be bounded by $n^{\vc}$ and combine this with \cref{lem:vc:weighted-fixed-distance-vector} to obtain our final contribution:
% For example, consider a star with pairwise distinct edge weights and let $M$ consist of its center.
% Each leaf then induces a different distance vector from $M$ to the corresponding singleton.
% Thus the number of candidates cannot be bounded solely in terms of $\vc$, but by $n^{\vc}$.

% With positive edge weights, the number of relevant distance vectors may depend on $n$.
% For example, consider a star with pairwise distinct edge weights and let $M$ consist of its center.
% Each leaf then induces a different distance vector from $M$ to the corresponding singleton.
% Thus the number of candidates cannot be bounded solely in terms of $\vc$, but by $n^{\vc}$.

\begin{theorem}\label{thm:vc:weighted-xp}
    {\WMCS} can be solved in time $\bO(2^\vc\cdot\vc\cdot n^{\vc+2})$.
\end{theorem}

\begin{proof}
    Consider an instance $(G,w,k,\col)$ of \WMCS.
    We first compute a minimum vertex cover $M$ of $G$ in time $\bO(2^\vc \cdot n)$ using the standard branching algorithm,
    where $|M|=\vc$.

    We first compute $\dist(u,v)$ for every $u\in M$ and $v\in V(G)$, in
    $\bO(\vc\cdot(m+n\log n))$ time using Dijkstra's algorithm.

    For each $u\in M$, choose a vertex $s_u\in V(G)$ and form the candidate vector
    \[
        \mathbf{d} \coloneq \bigl(\dist(u,s_u)\bigr)_{u\in M}.
    \]
    We discard the choice if some coordinate is infinite.
    There are $n^\vc$ choices, which can be enumerated in
    $\bO(n^\vc\cdot\vc)$ time using the precomputed distances.
    We allow the same vector to be generated more than once.

    Every relevant distance vector is generated: given a consistent subset $S$, choose each $s_u$ to be a closest vertex of $S$ to $u$.
    The resulting coordinate is then $\dist(u,s_u)=\dist(u,S)$.

    We apply \cref{lem:vc:weighted-fixed-distance-vector} to every retained candidate vector, reusing the precomputed distances.
    Since all edge weights are positive, the coordinate of $u$ is zero exactly when $s_u=u$.
    Hence there are at most $\binom{\vc}{g}(n-1)^{\vc-g}$ choices producing a vector with exactly $g$ zero entries.
    Each such vector is processed in
    $\bO(2^{\vc-g}\cdot\vc\cdot n^2)$ time.
    By the binomial theorem,
    \[
        \sum_{g=0}^{\vc}
            \binom{\vc}{g}(n-1)^{\vc-g}2^{\vc-g}
        =
         \sum_{g=0}^{\vc}
            \binom{\vc}{g}(2n-2)^{\vc-g} 1^{g}
        =
        (2n-2+1)^\vc
        \leq
        (2n)^\vc.
    \]
    The total running time is therefore
    $\bO(2^\vc\cdot\vc\cdot n^{\vc+2})$.
    This dominates both the distance preprocessing and the enumeration time.
    Since the candidates include the distance vector of an optimal consistent subset, returning the smallest solution found proves the theorem.
\end{proof}

%% file: Sections/conclusion.tex
\section{Concluding Remarks}
Our results provide a comprehensive understanding of the algorithmic upper and lower bounds for computing minimum consistent subsets. We note that while we focus on graphs, weighted graphs provide a more general setting for the problem than any metric space while for unweighted graphs our work completes the picture started in recent investigations of \MCSshort. One promising avenue for future work is to study approximation algorithms for computing consistent subsets. It would also be interesting to see how the foundational results obtained in our work correlate with the performance of empirical solvers on suitably structured instances.
% Open Questions \& Directions: \simon{discuss which of these we keep?}
% \begin{itemize}
%     \item Can we use the bucketing technique on the algorithm of Result 2? Perhaps even only on the unweighted setting?

%     \item What about the parameterization only by cutwidth?

%     \item Perhaps we can consider a more general case, where we are given explicitly a distance per pair of points.
%     What happens in this setting for say the parameterization by $k$ plus the different number of distances?

%     \item We can take a closer look at the motivation of the problem, and try to perhaps capture this motivation even better.
%     For instance, we can consider that there are also a number of vertices for which this constraint is not satisfied (\textit{faulty}).
%     Another interesting direction would be \emph{fairness}.

%     \item We believe that it shouldn't be too difficult to get $\NP$-hardness for the case where both $\Delta , c = \bO(1)$.

%     \item Can we perhaps show any inapproximability result based on the PIH?
%\end{itemize}